\documentclass[12pt]{article}

\usepackage{a4wide}
\usepackage{amssymb}
\usepackage{amsmath,amsthm}
\usepackage[latin1]{inputenc}
\usepackage{graphicx}
\usepackage{hyperref}
\usepackage{enumerate}
\usepackage{tikz}
\usepackage{tkz-graph}
\usetikzlibrary{shapes.geometric} 
\usepackage{mathrsfs}
\usepackage{verbatim}
\usepackage[ruled, linesnumbered, vlined,resetcount]{algorithm2e}
\usepackage{algorithmicx}
\usepackage{mathtools}
\usepackage{subfigure}
\usepackage{multirow}
\usepackage[normalem]{ulem}
\usepackage{array}
\newcolumntype{C}[1]{>{\centering\let\newline\\\arraybackslash\hspace{0pt}}m{#1}}
\usepackage{nccmath}

\newtheorem{definition}{Definition}

\newtheorem{proposition}{Proposition}
\newtheorem{remark}{Remark}

\hypersetup{colorlinks=true, linkcolor=blue, citecolor=blue, urlcolor=blue}

\newcommand{\agm}{{\small {\sf AGM}}}
\newcommand{\cagm}{{\small {\sf C-AGM}}}
\newcommand{\cagmdp}{{\small {\sf C-AGMDP}}}
\newcommand{\cpgm}{{\small {\sf CPGM}}}
\newcommand{\dcsbm}{{\small {\sf DCSBM}}}
\newcommand{\tricycle}{{\small {\sf TriCycle}}}
\newcommand{\cagmdpd}{{\small {\sf C-AGMDP-D}}}
\newcommand{\agmdptri}{{\small {\sf AGMDP-Tri}}}
\newcommand{\ca}[1]{\mathcal{#1}}
\newcommand{\tuple}[1]{\langle#1\rangle}

\newcommand{\mmid}{\!\mid\!}

\newcommand{\set}[1]{\{#1\}}
\newcommand{\itin}{{\it intra}}
\newcommand{\itbtw}{{\it inter}}

\newcommand{\itnew}{{\it new}}
\newcommand{\itprev}{{\it prev}}
\newcommand{\itcn}{{\it cn}}
\newcommand{\LS}{{\it LS}}
\newcommand{\CP}{{\ca{C}}}

\newcommand{\oln}{{\overline{n}}}

\newcommand{\olTheta}{{\overline{\Theta}}}
\newcommand{\nintratriangle}{{n^\itin_\triangle}}
\newcommand{\attrneigh}{\sim_{^{at}}}

\title{Publishing Community-Preserving Attributed Social Graphs 
with a Differential Privacy Guarantee}

\author{Xihui Chen$^1$, Sjouke Mauw$^{1,2}$ and Yunior Ram\'{i}rez-Cruz$^1$\\ 
{\small $^1$SnT, $^2$CSC, University of Luxembourg}\\ 
{\small 6, av. de la Fonte, L-4364 Esch-sur-Alzette, Luxembourg}\\ 
{\small \{xihui.chen, sjouke.mauw, yunior.ramirez\}\@@uni.lu}} 

\begin{document}
\maketitle

\begin{abstract} 
We present a novel method for publishing differentially private
synthetic attributed graphs. Unlike preceding approaches,
our method is able to preserve the community structure of the original
graph without sacrificing the ability to capture global structural properties.
Our proposal relies on \cagm, a new community-preserving generative model
for attributed graphs. We equip \cagm{} with efficient methods
for attributed graph sampling and parameter estimation.
For~the latter, we introduce differentially private
computation methods, which allow us to release community-preserving
synthetic \mbox{attributed} social graphs with a strong formal privacy guarantee.
Through comprehensive experiments, we show that our new model
outperforms its most relevant counterparts in synthesising 
differentially private attributed social graphs that preserve 
the community structure of the original graph, 
as well as degree sequences and clustering coefficients.
\end{abstract}


\section{Introduction}
\label{sec:intro}

The use of online social networks (OSNs) has grown steadily during
the last years, and is expected to continue growing in the future.
Billions of people share many aspects of their lives 
on OSNs and use these systems to interact with each other on a regular basis.
The ubiquity of OSNs has turned them into one of the most important
sources of data for the analysis of social phenomena.
Such analyses have led to significant findings used in a wide
range of applications, from efficient epidemic disease
control~\cite{MSWRH16,CRC15} to information diffusion~\cite{ZD11,KK17}.

Despite the undeniable social benefits that can be obtained from social network
analysis, access to such data by third parties such as researchers
and companies should understandably be limited
due to the sensitivity of the information stored in OSNs,
e.g.\ personal relationships, political preferences and religious affiliations.
In addition, the increase of public awareness about privacy and the
entry into effect of strong privacy regulations such as
GDPR~\cite{GDPR} strengthen the reluctance of OSN owners from
releasing their data. Therefore, it is of critical importance 
to provide mechanisms for privacy-preserving data publication 
to encourage OSN owners to release data for analysis.

\emph{Social graphs} are a natural representation of social networks, with nodes
corresponding to participants and edges to connections between participants.
In view of the privacy discussion, social network owners should only
release sanitised sample of the underlying social graphs.
However, it has been shown that even social graphs containing
only structural information remain vulnerable to privacy attacks
leveraging knowledge from public sources~\cite{NS09},
deploying sybil accounts~\cite{BDK11,MRT18c}, etc.
In order to prevent such attacks, a large number of graph
anonymisation methods have been devised. 
Initially, the proposed methods focused on editing
the original graph via vertex/edge additions and deletions until obtaining
a graph satisfying some privacy property.
A critical limitation of graph editing methods is their reliance
on assumptions about the adversary knowledge, which determine
the information that needs to be anonymised and thus the manner
in which privacy is enforced. To avoid this type of assumptions,
an increasingly popular trend is that of using semantic privacy notions,
which place formal privacy guarantees on the data processing algorithms
rather than the dataset. Among semantic privacy notions,
differential privacy~\cite{D06} has become the \emph{de facto} standard
due to its strong privacy guarantees.

According to the type of published data, we can divide differentially private
mechanisms for social graphs into two classes. 
The methods in the first category directly release specific statistics
of the underlying social graph, e.g.\ the degree sequence~\cite{KS12,HLMJ09}
or the number of specific subgraphs (triangles, stars, etc.)~\cite{ZCPSX15}.
The second family of methods focuses on publishing synthetic
social graphs as a replacement of real social networks
in a two-step process~\cite{MW09,SZWZZ11,WW13,XCT14}.
In the first step, differentially private methods are used to compute
the parameters of a generative graph model 
that accurately captures the original graph properties.
Then, in the second step, this model is sampled for synthetic graphs,
profiting from the fact that the result of post-processing 
the output of differentially private algorithms remains differentially 
private~\cite{KL10}. 

Differential privacy requires one to define a \emph{privacy budget} 
in advance, which determines the amount of perturbation that will be applied 
to the outputs of algorithms. In consequence, the methods in the first family 
need to either limit in advance the number of queries that will be answered 
or deliver increasingly lower quality answers. On the contrary, 
the methods in the second family can devote the entire privacy budget 
to the model parameter estimation, without further degradation 
of the privacy of the sampled graphs. For this reason, in this paper we focus 
on the second type of methods.

For analysts, the utility of synthetic graphs is determined
by the ability of the graph models to capture relevant properties
of the original graph. To satisfy this need, several graph models
have been proposed to accurately capture global structural properties
such as degree distributions and clustering coefficients,
as well as heterogeneous attributes
of the users such as gender, education or marital status.
A common limitation of the aforementioned approaches is their inability
to represent an important type of information: the community structure.
Informally, a community is a set of users who are substantially
more interrelated among themselves than to other users of the network.
This interrelation may, e.g., stem from the explicit existence of
relations between the users.
An example of such a community is a group of Gmail users who
frequently e-mail each other, as represented by the occurrence of a
large number of edges connecting the user nodes from the group.
Alternatively, interrelations may stem from the co-occurrence of
relevant features, such as users working at the same company or
alumni from the same university. The emergence of communities
has been documented to be an inherent property of social networks~\cite{PT19,YL15}.
For analysts, the availability of synthetic attributed graphs that preserve
the community structure of the original graph represents an opportunity
to improve existing applications. For example, they may be able to improve
online shopping recommendations based on the common purchases of users belonging
to the same community. Current models and methods are insufficient for enabling
such an analysis, as they either lack information about the community structure
or they lack vertex features. 

In this paper, we address the problem discussed in the previous paragraph
by introducing a new generative attributed graph model, \cagm{}
(short for \emph{Community-Preserving Attributed Graph Model}),
which in addition to global structural properties, is also capable
of preserving the community structure of the original graph.
\cagm{} is based on the attributed graph model \agm~\cite{PMFNG14},
and improves on it by incorporating the capability of preserving 
the number and sizes of the communities of the original graph, as well as 
the densities of intra- and inter-community connections
(that is, connections between nodes belonging to the same community or
to different communities, respectively). \cagm{} also preserves 
a number of statistics describing the correlations between the feature 
vectors that describe the users and the existence of connections between 
pairs of users, as well as their community co-affiliation.
We equip \cagm{} with efficient parameter estimation
and graph sampling methods, and provide differentially private variants
of the former, which allow us to release synthetic attributed social graphs
with a strong privacy guarantee and increased utility
with respect to preceding approaches.

\vspace{1mm}
\noindent
{\bf Summary of contributions:} 
\begin{itemize}
\item We propose a new generative attributed graph model, \mbox{\cagm},
which captures a number of properties of the community structure, as discussed
in the previous paragraph, along with global structural properties.
\item We present efficient methods for learning an instance of our model
from an input graph and sampling community-preserving synthetic attributed graphs
from this instance. We show, via a number of experiments
on real-world social networks, that the community structures of synthetic graphs
sampled from our model are more similar to those of the original graphs
than those of the graphs sampled from previously existing models.
Additionally, we show that this behaviour is obtained without sacrificing
the ability to preserve global structural features.
\item We devise differentially private methods for computing the parameters
of the new model. We demonstrate that our methods are practical
in terms of efficiency and accuracy. To support the latter claim,
we empirically show that differentially private synthetic
attributed graphs generated by our model suffer a reasonably low degradation
with respect to their counterparts, in terms of their ability
to capture the community structure and structural features
of the original graphs.
\end{itemize}

\section{Related Work}
\label{sec:related work}

\noindent
{\bf Private graph synthesis.}
The key to synthesising social graphs is the model which determines
both the information embedded in the published graphs and the properties preserved.
Mir et al.~\cite{MW09} used the Kronecker graph generative model~\cite{LF07}
to generate differentially private graphs. As the Kronecker model cannot
accurately capture structural properties,
Sala et al.~\cite{SZWZZ11} proposed an alternative approach which makes use
of the $dK$-graph model. Wang et al.~\cite{WW13} further improved
the work of Sala et al. by considering global sensitivity instead of
local sensitivity (refer to Section~\ref{sec:preliminaries}
for the definition of sensitivity).
Xiao et al.~\cite{XCT14} introduced the {\small {\sf HRG}}-graph model~\cite{CN08}
and found that it can further reduce the amount of added noise
and thus increase the accuracy.

The approaches described so far work on unlabelled graphs.
Pfiffer et al.~\cite{PMFNG14} introduced a new model called \agm,
which attaches binary attributes to nodes and captures the correlations
between shared attributes and the existence of connections.
Jorgensen et al.~\cite{JYC16} adopted this model and proposed
differentially private methods to accurately estimate the model parameters.
They also designed a new graph generation algorithm based on the {\small {\sf TCL}}
model~\cite{PFMN12}, which enables the model to sample attributed graphs preserving
the clustering coefficient. As discussed previously, \cagm,
the model introduced in this paper, is comparable to this model
in preserving global structural properties of the original graphs,
but it outperforms it by also capturing the community structure.

\vspace{1mm}\noindent
{\bf Private statistics publishing.} Degree sequences and degree correlations
are two types of the statistics frequently studied in the literature.
The general trend in publishing these statistics under differential privacy
consists in adding noise to the original sequences and then post-processing
the perturbed sequences to enforce or restore certain properties,
such as graphicality~\cite{KS12}, vertex order 
in terms of degrees~\cite{HLMJ09}, etc. 
Subgraph count queries, e.g. the number of triangles or $k$-stars,
have also received considerable attention. Among the approaches
to accurately compute such queries, we have the so-called
ladder functions~\cite{ZCPSX15} and smooth sensitivity~\cite{KRSY14,WWZX12}.

\vspace{1mm}\noindent
{\bf Community-preserving graph generation models.}
A number of existing random graph models claim to capture
community structure, e.g., {\small {\sf BTER}}~\cite{KPPS14},
{\small {\sf ILFR}}~\cite{PT19}, {\small {\sf SBM}}~\cite{HLL83}
and its variants (e.g., \dcsbm~\cite{KE11} and {\small {\sf DCPPM}}~\cite{N16}).
{\small {\sf BTER}} generates community-preserving social graphs 
given expected node degrees and, for every degree value $\sigma$, 
the average of the clustering coefficients
of the nodes of degree $\sigma$. The model assumes that every community
is a set of $\sigma$ nodes with degree $\sigma$. On the contrary, \cagm{}
makes no assumptions on the community partition received.
Finally, {\small {\sf ILFR}} and the variants of {\small {\sf SBM}}
preserve edge densities at the community level but, unlike our new model,
they do not preserve the clustering coefficients of the original graph.

\section{Preliminaries}
\label{sec:preliminaries}

\subsection{Notation}
\label{ssec:our model}

An attributed graph is represented as a triple $G=(\ca{V},\ca{E},X)$,
where $\ca{V}=\{v_1,v_2,\ldots, v_n\}$ is the set of nodes,
$\ca{E}\subseteq \ca{V}\times \ca{V}$ is the set of edges, and $X$
is a binary matrix called the \emph{attribute matrix}.
The $i$-th row of $X$ is the attribute vector of $v_i$,
which is individually denoted by $\tau(v_i)$.
Every column of $X$ represents a binary feature,
which is set to $1$ (\textbf{true}), or $0$ (\textbf{false}), for each user.
For example, if the $j$-th column represents the attribute ``\emph{owning a car}'',
$X_{ij}=1$ means that the user represented by~$v_i$ owns a car.
Non-binary real-life attributes are assumed to be binarised.
For example, a binarisation of the integer-valued
attribute ``\emph{age}'' is $\{$``\emph{age $\le 16$}'',
``\emph{$17\le$ age $\le 26$}'', ``\emph{$27\le$ age $\le 64$}'',
``\emph{age $\ge 65$}''$\}$. The order of the columns of $X$
is fixed, but arbitrary, and has no impact on the results
described hereafter. Throughout the paper,
we deal with undirected graphs. That is, if $(v_i, v_j)\in \ca{E}$,
then $(v_j,v_i)\in \ca{E}$. Additionally, we use $A$ to denote
the adjacency matrix of the graph.

We use $\ca{C}=\set{C_0,C_1,\ldots, C_p}$, with $C_i\subseteq \ca{V}$
for every $i\in\{0,1,\ldots,p\}$, to represent
a \emph{community partition} of the attributed graph.
As the term suggests, in this paper we assume that
$C_i\cap C_j=\emptyset$, with $0\le i < j \le p$,
and $\displaystyle\cup_{C_i\in\ca{C}}C_i=\ca{V}$.
The community $C_0$ has a special interpretation. Since some community detection
algorithms assign no community to some vertices, we will use $C_0$
as a ``discard'' community of unassigned vertices.
We do so to avoid having a potentially large number of singleton communities,
for which no meaningful co-affiliation statistics can be computed.
We use $\psi_\ca{C}(v_i)$ to denote the community to which the node $v_i$
belongs in the community partition $\ca{C}$.
We will use $\psi(v_i)$ for short in cases where the partition
is clear from the context.

\subsection{Differential Privacy}
\label{ssec:differential privacy}

\emph{Differential privacy}~\cite{D06} is a well studied statistical notion
of privacy. The intuition behind it is to randomise the output of an algorithm
in such a way that the presence of any individual element in the input dataset
has a negligible impact on the probability of observing any particular output.
In other words, a mechanism is $\varepsilon$-differentially private
if for any pair of \emph{neigh\-bouring datasets}, i.e. datasets that
only differ by one element, the probabilities of obtaining any output
are measurably similar.
The amount of similarity is determined by the parameter $\varepsilon$,
which is commonly called the \emph{privacy budget}.
In what follows, we will use the notation $\ca{D}$ for the set
of possible datasets, $\ca{O}$ for the set of possible outputs,
and $D\sim D'$ for a pair of neighbouring datasets.

\begin{definition}[$\varepsilon$-differential privacy \cite{D06}]
A randomised mechanism $\ca{M}\colon\ca{D}\to\ca{O}$ satisfies
$\varepsilon$-differential privacy if for every pair
of neighbouring datasets $D,D'\in\ca{D}$, $D\sim D'$,
and for every $S\subseteq\ca{O}$, we have
\[
\Pr(\ca{M}(D)\in S)
\leq e^\varepsilon \Pr(\ca{M}(D')\in S).
\]
\end{definition}

A number of differentially private mechanisms have been proposed.
For queries of the form $q\colon\ca{D}\to \mathbb{R}^n$,
the most widely used mechanism to enforce differential privacy
is the so-called \emph{Laplace mechanism}, which consists
in obtaining the (non-private) output of $q$ and adding to every component
a carefully chosen amount of random noise, which is drawn from
the Laplace distribution 
$$Lap(\lambda)\colon f(y \mmid\lambda)=\frac{1}{2\lambda}
\exp(\frac{-\mmid y\mmid}{\lambda}),$$
where $y$ is a real-valued variable indicating the noise to be added,
$\lambda=\frac{\Delta_q}{\varepsilon}$
and $\Delta_q$ is a property of the original function $q$
called \emph{global sensitivity}. This property is defined
as the largest difference between the outputs of $q$ for any pair
of neighbouring datasets,
that is $$\Delta_q=\max_{D\sim D'}\left\Vert q(D)-q(D')\right\Vert_1,$$
where $\left\Vert\cdot\right\Vert_1$ is the $L_1$ norm.
For categorical (non-numerical) queries of the form $q\colon\ca{D}\to \ca{O}$,
where $\ca{O}$ is a finite set of categories, the so-called
\emph{exponential mechanism}~\cite{MST07} is the most commonly used.
In this case, for each value $o\in\ca{O}$, a score is assigned
by a function (usually called \emph{scoring function})
quantifying the value's utility, denoted by $u(o,D)$.
The global sensitivity of $u$ is
$$\Delta_u=\max_{o\in\ca{O},D\sim D'}\lvert u(o,D)-u(o,D')\rvert,$$ 
and the randomised output is drawn with probability 
$\frac{\exp(\frac{\varepsilon\cdot u(o,D)}{2\Delta_u})}
{\sum_{o'\in\ca{O}} \exp(\frac{\varepsilon\cdot u(o',D)}{2\Delta_u})}.$

Differentially private methods are composable~\cite{MS10}.
That~is, given a set of algorithms $\{\ca{M}_1, \ca{M}_2, \ldots, \ca{M}_n\}$
such that $\ca{M}_i$ ($1\leq i\leq n$) satisfies
$\varepsilon_i$-differential privacy, if~the algorithms are applied
sequentially and the results combined by a deterministic
method, then the final result satisfies $\sum_i\varepsilon_i$-differential privacy.
If the algorithms are applied independently on disjoint subsets of the
input, then $\max_i\{\varepsilon_i\}$-differential privacy is satisfied.
Moreover, post-processing on the output of an $\varepsilon$-differentially
private algorithm also satisfies $\varepsilon$-differential privacy
if the post-processing is deterministic or randomised with a source of randomness
independent from the noise added to the original algorithm~\cite{KL10}.
These properties allow us to divide a complex computation, such as the set of
model parameters in our case, into a sequence of sub-tasks for which
differentially private methods exist or can be more easily developed.

In addition to the global sensitivity, a dataset-dependent notion,
called \emph{local sensitivity}~\cite{NRS07}, has been enunciated.
The local sensitivity of query $q$ on a dataset $D$ is defined~as 
$$\LS_q(D) =\max_{D\sim D'}\left\Vert q(D)-q(D')\right\Vert_1,$$ that is, 
the maximum difference between the output of $q$ on $D$ and those 
on its neigh\-bouring datasets. 
It is simple to see that \mbox{$\Delta_q=\max_D \LS_q(D)$}. 

\section{The {\sf C-AGM} Model}
\label{sec:C-AGM}

In this section we give the formal definition of \cagm.
We introduce the methods for sampling synthetic graphs from the model,
and describe the methods for learning the model parameters
from an attributed graph.

\subsection{Overview}
\label{ssec:overview}

Algorithm~\ref{alg:dpworkflow} summarises the process by which \cagm{}
is used for publishing synthetic attributed graphs.
As discussed in~\cite{HLMJ09,KL10,KS12}, synthetic graph generation 
is done as a post-processing step of the differentially private computation, 
so the synthetic graphs are also differentially private. 

\begin{algorithm}[!ht]
\caption{Given $G=(\ca{V},\ca{E},X)$,
obtain $t$ differentially private attributed synthetic graphs.
}
\label{alg:dpworkflow}
\SetAlgoLined
Split privacy budget\;
Obtain differentially private community partition\;
Differentially-privately estimate \cagm{} parameters\;
\For{$i\in\{1,2, \ldots, t\}$}{
Sample $X_i$ from \cagm\;
Sample $\ca{E}_i$ from \cagm\;
$G_i\gets(\ca{V},\ca{E}_i,X_i)$
}
\end{algorithm}

The manner in which the privacy budget is split
among the different computations (step~1)
is discussed in Section~\ref{sec:DP-calculation}.
For the differentially private community partition (step~2),
we introduce in this paper an extension of the algorithm 
ModDivisive~\cite{NIR16}. The purpose of this extension is to incorporate 
information from node attributes into the objective function optimised 
by ModDivisive. We discuss the community  partition method in detail
in Section~\ref{ssec:dp-community detection}.
A thorough description of the parameters of \cagm{} is given
in Section~\ref{ssec:C-AGM},
and parameter estimation is discussed in Section~\ref{ssec:calculate C-AGM}.

Once the model parameters have been estimated, we can sample any number
of synthetic attributed graphs from the model, as described in steps~4 to~8
of Algorithm~\ref{alg:dpworkflow}. The differentially private parameter estimation 
methods introduced in this paper use the notion 
of \emph{neighbouring attributed graphs}~\cite{JYC16}, which is discussed 
in detail in the preamble of Section~\ref{sec:DP-calculation}. 
\mbox{Under} this notion, the existence of relations (edges) 
and personal characteristics of the network users (feature vectors)
are treated as sensitive, but vertex identities are not. 
Thus, the synthetic graphs generated by Algorithm~\ref{alg:dpworkflow} 
have the same vertex set as the original graph, whereas the attribute matrix
and the edge set are sampled from the model (step~7).
For every new synthetic attributed graph,
we first sample the attribute matrix, and then this matrix is used,
in combination with an edge generation model (Section~\ref{sssec:CPGM}),
to generate the edge set of the synthetic graph.
There are two reasons for dividing this process
into two steps. The first one is to make the sampling process efficient.
The second reason is to profit from the two-step process to enforce
the intuition that users with similar features are more likely 
to be connected in the social network.
The attributed graph sampling procedure is discussed in detail 
in Section~\ref{ssec:SampleCAGM}. 

\subsection{Model Parameters}
\label{ssec:C-AGM}

As we discussed in Section~\ref{sec:intro}, given an attributed graph $G$
and a community partition $\ca{C}$ of $G$, the purpose of \cagm{}
is to capture a number of properties of $\ca{C}$ that are overlooked
by previously defined models, without sacrificing the ability to capture
global structural properties such as degree distributions
and clustering coefficients. To that end, \cagm{} models the following
properties of the community partition:

\begin{enumerate}
\item the number and sizes of communities;
\item the number of intra-community edges in every \mbox{community};
\item the number of inter-community edges;
\item the distributions of attribute vectors in every community;
\item the distributions of the so-called \emph{attribute-edge correlations}
\cite{JYC16}, for the set of inter-community edges
and for the set of intra-community edges in every community.
\end{enumerate}

Graphs generated by \cagm{} will have the same number
of vertices as the original graph, as well as the same number of communities.
Moreover, every community will have the same cardinality as in the
original graph, and the same number of intra-community edges. The
number of inter-community edges of the generated graph will also be
the same as that of the original graph. Notice that the model preserves
the total number, but not necessarily the pairwise numbers
of inter-community edges for every pair of communities.

\emph{Attribute-edge correlations} were defined in~\cite{JYC16}
as heuristic values for characterising the relation between the feature vectors
labelling a pair of vertices and the likelihood that these vertices
are connected. They encode the intuition that, for example,
co-workers who attended the same university and live near to each other
are more likely to be friends than persons with fewer features in common,
whereas friends are more likely to support the same sports teams
or go to the same bars than unrelated persons. In~\cite{JYC16},
attribute-edge correlations are considered to behave uniformly
over the entire graph. Here, we introduce the rationale
that they behave differently within different communities,
as well as across communities.

A key element in the representation of attribute-edge correlations
is the notion of \emph{aggregator functions}.
An aggregator function $\beta\colon\{0,1\}^k\times\{0,1\}^k\to \ca{B}$
maps a pair of attribute vectors $x,x'$ of dimensionality $k$
into a value in a discrete range $\ca{B}$, which is used as a descriptor,
also called \emph{aggregated feature}, of the pair $(x,x')$.
For example, $\ca{B}$ can contain a set of similarity levels for pairs
of feature vectors, such as \{\emph{low}, \emph{medium}, \emph{high}\},
and $\beta$ can map a pair of vectors whose cosine similarity
is in the interval $[0,0.33]$ to \emph{low}, a pair of vectors
whose cosine similarity is in the interval $[0.67,1]$ to \emph{high}, etc.
Attribute-edge correlations, along with the community-wise distributions
of attribute vectors, are useful for analysts, as they allow to characterise
the members of a community in terms of frequently shared features,
hypothesise explanations for the emergence of a community, etc.

Formally, a \cagm{} model is defined as a quintuple
$\tuple{\ca{V}, \ca{C}, \Theta_M^c, \Theta_X^c, \Theta_F^c}$,
where:
\begin{itemize}
\item {\bf $\ca{V}$} is a set of vertices.
\item {\bf $\ca{C}$} is a community partition of $\ca{V}$.
\item {\bf $\Theta_M^c$} is an instance of an edge set generative model
that preserves properties 1 to 3 of the community partition $\ca{C}$,
as well as degree distributions and clustering coefficients.
The model introduced in this paper is called \cpgm,
and is described in detail in Section~\ref{sssec:CPGM}.
\item {\bf $\Theta_X^c$} is an instance of an attribute vector generative model,
which aims to preserve property 4. The model defines,
for every community $C\in\ca{C}$ and every attribute vector $x$,
the probability $\Pr(\tau(v)=x\mmid v\in C, \Theta_X^c)$
that a vertex in $C_i$ is labelled with $x$.
The model introduced in this paper is described in detail
in Section~\ref{sssec:estimating-thetacX}.
\item {\bf $\Theta_F^c$} is an instance of a generative model for
attribute-edge correlations, which aims to preserve property 5.
This model defines:
\begin{itemize}
\item The discrete range $\ca{B}$ and an aggregator function $\beta$.
\item The probability
$$\Pr(\beta(\tau(v_i),\tau(v_j))=s\mmid\Theta_F^c,
\psi_\CP(v_i)=\psi_\CP(v_j)=C,A_{i,j}\!=\!1)$$
for every community $C\in\ca{C}$ and every value $s\in\ca{B}$.
\item The probability
$$\Pr(\beta(\tau(v_i),\tau(v_j))=s\mmid\Theta_F^c,
\psi_\CP(v_i)\neq\psi_\CP(v_j),A_{i,j}\!=\!1)$$
for every value $s\in\ca{B}$.
\end{itemize}
The instantiations that we propose for these three components
are described in detail in Section~\ref{sssec:estimating-thetaFc}.
\end{itemize}

\subsection{Sampling Attributed Graphs from an Instance of {\sf C-AGM}}
\label{ssec:SampleCAGM}

Given a \cagm{} model
$\ca{G}=\tuple{\ca{V}, \ca{C}, \Theta_M^c, \Theta_X^c, \Theta_F^c}$,
with $\ca{V}=\{v_1,v_2,\ldots,v_n\}$,
an attributed graph $G=(\ca{V},\ca{E},X)$ is sampled from $\ca{G}$
with probability $\Pr(G\mid\ca{G})=\Pr(\ca{E},X\mid\ca{G})$ which,
for the sake of tractability, is approximated as
$$\Pr(\ca{E},X\mmid\Theta_F^c, \Theta_X^c,\CP,\Theta_M^c)=
\Pr(\ca{E}\mmid\Theta_F^c,\Theta_M^c, X,\CP)\cdot \Pr(X\mmid\Theta_X^c, \CP).$$
That is, we first sample from $\Theta_X^c$ the attribute vectors labelling
each vertex and then use them in sampling the edge set.
Again, to keep the sampling process tractable,
we introduce an additional independence assumption, according to which
\[
\Pr(X\mmid\Theta_X^c,\CP)=\prod_{v\in\ca{V}}\Pr(\tau(v)\mmid\psi_\CP(v)).
\]
The computation of the probabilities of the form $\Pr(x\mmid \psi_\CP(v))$
will be discussed in Sect\-ion~\ref{sssec:estimating-thetacX}.
Introducing the assumption that edges are sampled independently
from each other, the probability of generating $\ca{E}$
given $\Theta_F^c$, $\Theta_M^c$, $X$, and $\CP$ is
$$\Pr(\ca{E}\mid\Theta_F^c,\Theta_M^c, X,\CP)=\\
\prod_{v_i,v_j\in\ca{V}}
\Pr(A_{i,j}\mmid \Theta_F^c,\Theta_M^c, {\beta(\tau(x_i), \tau(x_j))},\CP).$$

As it is inefficient to sample edges directly from this
distribution, we adapt the sampling method introduced in~\cite{PMFNG14}
to account for the computation of community-wise separated counts.
Thus, edges are drawn from the distribution
\[
Q(i,j) \propto  Q'_M(i,j)\cdot \Gamma(\beta(\tau(v_i), \tau(v_j)), \CP),
\]
where $Q'_M(i,j)$ is the probability that $(v_i,v_j)$ is drawn
from the edge generation model $\Theta_M^c$, given $\CP$, as a candidate edge;
while $\Gamma(\beta(\tau(v_i),\tau(v_j)), \CP)$
is the probability that it is accepted by $\Theta_F$, given $\CP$.
We split the computation of $\Gamma(\beta(\tau(v_i),\tau(v_j)), \CP)$ into
two cases: $\Gamma_{\itin}(\beta(\tau(v_i),\tau(v_j)), C)$, for every $C\in\CP$
and every $i,j$ such that $\psi_\CP(v_i)=\psi_\CP(v_j)=C$;
and $\Gamma_{\itbtw}(\beta(\tau(v_i),\tau(v_j)))$,
for every $i,j$ such that $\psi_\CP(v_i)\neq\psi_\CP(v_j)$.
Formally, we have $$Q'_M(i,j) = \frac{\Pr(A_{i,j}=1\mid\Theta_M^c,\CP)}
{\sum_{v_p, v_q\in\ca{V}}\Pr(A_{p,q}=1\mid\Theta_M^c,\CP)},$$
\begin{displaymath}
\Gamma_{\itin}(\beta(\tau(v_i),\tau(v_j)), C)=
\frac{R_{\itin}(\beta(\tau(v_i),\tau(v_j)), C)}
{SupR},
\end{displaymath}
\begin{displaymath}
\text{and }\Gamma_{\itbtw}(\beta(\tau(v_i),\tau(v_j)))=
\frac{R_{\itbtw}(\beta(\tau(v_i),\tau(v_j)))}{SupR},
\end{displaymath}
where
$$R_{\itin}(\beta(\tau(v_i),\tau(v_j)),C)=\frac{\Pr(\beta(\tau(v_i),\tau(v_j))
\mid \Theta_F^c,\psi_\CP(v_i)=\psi_\CP(v_j)=C,A_{i,j}=1)}
{\Pr(\beta(\tau(v_i),\tau(v_j))\mid \Theta_M^c,
\psi_\CP(v_i)=\psi_\CP(v_j)=C,A_{i,j}=1)},$$
$$R_{\itbtw}(\beta(\tau(v_i),\tau(v_j)))=
\frac{\Pr(\beta(\tau(v_i),\tau(v_j))\mid \Theta_F^c,
\psi_\CP(v_i)\neq\psi_\CP(v_j),A_{i,j}=1)}
{\Pr(\beta(\tau(v_i),\tau(v_j))\mid \Theta_M^c,
\psi_\CP(v_i)\neq\psi_\CP(v_j),A_{i,j}=1)},$$
\begin{displaymath}
\text{and }SupR=\sup \bigcup_{s\in\ca{B}, C\in\CP}\left(R_{\itin}(s,C)\cup
R_{\itbtw}(s)\right).
\end{displaymath}

The computation of $Q'_M(i,j)$ will be discussed in Section~\ref{sssec:CPGM},
whereas that of\\ $\Gamma_{\itbtw}(\beta(\tau(v_i),\tau(v_j)))$,
and every $\Gamma_{\itin}(\beta(\tau(v_i),\tau(v_j)), C)$
will be discussed in Section~\ref{sssec:estimating-thetaFc}.

Algorithm~\ref{alg:samplingCAGM} describes the procedure to sample
an attributed graph from \cagm.
The method first generates the attribute vectors (line~1).
Then, it pre-computes the acceptance probabilities (lines~2 to~11).
In line~3, the call to {\small {\sf SampleEdgeSet}} consists in the sequential
execution of Algs.~\ref{alg:gencommgraph} and~\ref{alg:gengraphtri},
which will be described in detail in Section~\ref{sssec:CPGM}.
Finally, the loop in lines~12 to~20 repeatedly draws candidate edges
from the edge generation model and adds to the graph those that are accepted
according to the pre-computed probabilities (lines~17 and~18).
The method stops when the required number of edges is added.

\begin{algorithm}[!ht]
\caption{${\sf SampleFromCAGM}(\ca{V}, \CP, \Theta^c_M, \Theta_X^c, \Theta_F^c)$}
\label{alg:samplingCAGM}
\SetAlgoLined
{\footnotesize
$X' \gets {\sf SampleAttributeVectors}(\Theta_X^c)$\;
$Q'_M\gets {\sf ComputeQM}(\Theta_M^c, \CP)$\;
$\ca{E}' \gets {\sf SampleEdgeSet}(Q'_M)$\;
\For{$s\in\ca{B}$}{
	Compute $\Gamma_\itbtw(s)$\;
	\For{$C\in\CP$}{
		Compute $\Gamma_\itin(s,C)$
	}
}
$\ca{E}'\gets\emptyset$\;
\While{$|\ca{E}'|<|\ca{E}|$}{
    $(v,w) \gets {\sf SampleEdge}(Q'_M)$\; 
    $s\gets\beta(\tau(v),\tau(w))$\;
    $u\gets {\sf Uniform}(0,1)$\;
    \If{$(\psi_\CP(v)=\psi_\CP(w)\land u\le\Gamma_\itin(s,\psi_\CP(v))
    \mbox{~\textbf{or}~} (\psi_\CP(v)\neq\psi_\CP(w)\land u\le \Gamma_\itbtw(s)$}{
        $\ca{E}'\gets \ca{E}'\cup \{(v,w)\}$\;
    }
}
return $X',\ca{E}'$\;
}
\end{algorithm}

\subsubsection{Edge generation model}
\label{sssec:CPGM}

As we discussed in Section~\ref{ssec:C-AGM}, the component $\Theta_M^c$
of \mbox{\cagm} is an edge generation model which preserves
several properties of the community partition of the original graph
(properties~1 to~3 listed in Section~\ref{ssec:C-AGM}),
in addition to the degree distribution and clustering coefficients.
We call this model \cpgm, and describe it in what follows.

The model takes as input the set of vertices, as well as the expected number
of neighbours of every vertex $v$ within its community
(that is, its \emph{intra-community degree}, denoted by $d_\itin(v)$)
and the expected number of neighbours outside its community
(that is, the \emph{inter-community degree}, denoted by $d_\itbtw(v)$). 
These values are used to enforce the expected densities within every community
and between communities. Additionally, adapting to our setting
a heuristics introduced in~\cite{JYC16}, the model also requires the number
of triangles having all vertices in one community
(which we call \emph{intra-community triangles} and denote by $n_\triangle^\itin$),
as well as the number of triangles spanning more than one community
(\emph{inter-community triangles}, denoted by $n_\triangle^\itbtw$).
As shown empirically in~\cite{JYC16}, synthetic graphs that preserve
the number of triangles of the original graph are more likely to approximate
the clustering coefficient of the original graph. We adopt this intuition 
as well, but unlike \cite{JYC16}, we separate the counts
of intra- and inter-community triangles. As we will discuss
in Section~\ref{sec:DP-calculation}, $n_\triangle^\itin$ and $n_\triangle^\itbtw$
can be efficiently and accurately computed under differential privacy. 

According to our model, the edge sampling process consists of two steps.
The first step generates a graph that preserves the intra- and inter-community
degrees, but not the number of intra- and inter-community triangles.
Then, the second step iteratively edits the original edge set until
$n_\triangle^\itin$ and $n_\triangle^\itbtw$ are enforced.

At the first step, we follow the idea of the {\small {\sf CL}}~model~\cite{CL02}.
For every pair of vertices $v$ and $w$ satisfying
$\psi_{\mathcal{C}}(v)=\psi_{\mathcal{C}}(w)=C$, the intra-community edge
$(v,w)$ is added with probability
$\pi^\itin_C(v,w)=\frac{d_\itin(v)d_\itin(w)}{2m^\itin_C}$, 
where $m^\itin_C$ is the original number of intra-community edges in~$C$.
That is, intra-community edges are added with a probability proportional
to product of the intra-community degrees of the linked vertices.
If $\psi_{\mathcal{C}}(v)\ne\psi_{\mathcal{C}}(w)$, then the inter-community edge
$(v,w)$ is added with probability
$\pi^\itbtw(v,w)=\frac{d_\itbtw(v)d_\itbtw(w)}{2m^\itbtw}$, where $m^\itbtw$ 
is the total number of inter-community edges in the original graph.
Algorithm~\ref{alg:gencommgraph} describes the first step of the generation process.

\begin{algorithm}[!ht]
\caption{{\sf GenInitialEdgeSet}$(d_\itin, d_\itbtw,\ca{C})$}
\label{alg:gencommgraph}
\SetAlgoLined
{\footnotesize
$\ca{E}\gets\emptyset$\;
\For{$C\in \ca{C}$}{
    $m^\itin_C \gets\frac{1}{2}{\sum_{v\in C}d_\itin(v)}$\;
    $m\gets 0$\;
    \While{$m\leq m^\itin_C$}{
        $(v,w)\gets\mbox{\sf Sample}(\pi_C^\itin)$\;
        \If{$(v,w)\notin \ca{E}$}{
            $\ca{E}\gets \ca{E}\cup \{(v,w)\}$\;
            $m\gets m+1$\;
        }
    }
}
$m^\itbtw\gets\frac{1}{2}{\sum_{v\in\ca{V}}d_\itbtw(v)}$\;
\While{$m\leq\sum_{C\in\ca{C}}m^\itin_C+m^\itbtw$}{
    $(v,w)\gets\mbox{\sf Sample}(\pi^\itbtw)$\;
    \If{$(v,w)\notin \ca{E}$}{
        $\ca{E}\gets \ca{E}\cup \{(v,w)\}$\;
        $m\gets m+1$        \;
    }
}
}
\end{algorithm}

\begin{algorithm}[!ht]
\caption{{\sf GetFinalEdgeSet}$(d_\itin, d_\itbtw, n_\triangle^\itin,$
$n_\triangle^\itbtw, \CP)$}
\label{alg:gengraphtri}
\SetAlgoLined
{\footnotesize
$\mu_\triangle^\itin\gets {\sf CountIntraCommTriangles}(\ca{E}$)\;
\While{$\mu_\triangle^\itin<n_\triangle^\itin$}{
    Uniformly sample $C$ from $\CP$\;
    Sample $v_1$ from $C$ with probability $\frac{d_\itin(v_1)}{2m^\itin_C}$\; 
    Uniformly sample $v_2$ from $\ca{N}_\itin(v_1)$\;
    Uniformly sample $v_3$ from $\ca{N}_\itin(v_2)$\;
    \If{$(v_1,v_3)\not\in\ca{E}\land v_3\neq v_1$}{
        $(v_1', v_2')\gets {\sf GetOldestIntraCommEdge}(\ca{E},\ca{C})$\;
        $n_\itcn^\itprev\gets {\sf GetCommonNeighbour}(v_1', v_2')$\;
        $\ca{E}\gets\ca{E}/\{(v_1', v_2')\}$\;
        $n_\itcn^\itnew\gets {\sf GetCommonNeighbour}(v_1, v_3)$\;
        \eIf{$n_\itcn^\itprev<n_\itcn^\itnew$}{
            $\ca{E}\gets \ca{E}\cup\{(v_1, v_3\}$\;
            $\mu^\itin_\triangle
            \gets \mu^\itin_\triangle- n_\itcn^\itprev+n_\itcn^\itnew$\;
        }{
            $\ca{E}\gets\ca{E}\cup\{(v_1', v_2')\}$\;
        }
    }
}
$\mu_\triangle^\itbtw\gets {\sf CountInterCommTriangles}(\ca{E}$)\;
\While{$\mu_\triangle^\itbtw<n_\triangle^\itbtw$}{
    Sample $v_1$ from $\ca{V}$ with probability $\frac{d_\itbtw(v_1)}{2m^\itbtw}$\; 
    Uniformly sample $v_2$ from $\ca{N}_\itbtw(v_1)$\;
    Uniformly sample $v_3$ from $\ca{N}_\itin(v_2)$\;
    $(v_1', v_2')\gets {\sf GetOldestInterCommEdge}(\ca{E},\ca{C})$\;
    $n_\itcn^\itprev\gets {\sf GetCommonNeighbour}(v_1', v_2')$\;
    $\ca{E}\gets \ca{E}/\{(v_1', v_2')\}$\;
    $n_\itcn^\itnew\gets {\sf GetCommonNeighbour}(v_1, v_3)$\;
    \eIf{$n_\itcn^\itprev<n_\itcn^\itnew$}{
        $\ca{E}\gets \ca{E}\cup\{(v_1, v_3)\}$
        $\mu^\itbtw_\triangle
        \gets \mu^\itbtw_\triangle - n_\itcn^\itprev+n_\itcn^\itnew$\;
    }{
        $\ca{E}\gets\ca{E}\cup\{(v_1', v_2')\}$\;
    }
}
}
\end{algorithm}

At the second step, we use the intuition that the clustering behaviour in
social networks stems from the higher likelihood of users with common friends
to connect~\cite{PFMN12}, thus creating triangles.
Algorithm~\ref{alg:gengraphtri} enforces the values of $n_\triangle^\itin$
and $n_\triangle^\itbtw$ of the original graph on the graph synthesised
by Algorithm~\ref{alg:gencommgraph}. In Algorithm~\ref{alg:gengraphtri}, 
we denote by $\ca{N}_\itin(v)$ the set of neighbours of $v$ in its community, 
that is $\ca{N}_\itin(v)=\{w\;|\;\psi_\CP(v)=\psi_\CP(w)\land (v,w)\in\ca{E}\}$. 
Likewise, we denote by $\ca{N}_\itbtw(v)$ the set of neighbours 
of $v$ in different communities, 
that is $\ca{N}_\itbtw(v)=\{w\;|\;\psi_\CP(v)\ne\psi_\CP(w)\land (v,w)\in\ca{E}\}$. 
In Algorithm~\ref{alg:gengraphtri}, 
$n_\triangle^\itin$ is enforced first because adding or removing 
an intra-community edge may change the number of inter-community triangles as well,
whereas inter-community triangles can be created without modifying the number
of intra-community triangles. 
At every iteration, we sample a new edge. If replacing the oldest intra-community
edge (in terms of the order of creation by Algorithm~\ref{alg:gencommgraph})
with the newly sampled edge causes the number of intra-community triangles
to increase, we make the edge exchange permanent.
Otherwise, we do not add the newly sampled edge and set the oldest edge 
to be the youngest, keeping it in the graph. 
The iteration stops when the number of intra-community triangles is greater than
or equal to that of the original graph. Then, we proceed to enforce
the number of inter-community triangles by adding inter-community edges.
In this case the idea is to find open ``wedges'' composed
of one intra-community edge $(u,v)$
and one inter-community edge $(v,w)$ such that the edge $(u,w)$ has not been
added to the graph. This ensures that newly added edges will not affect
the number of intra-community triangles. Let $(v',w')$ be the oldest
inter-community edge. If the graph obtained by removing $(v',w')$
and adding $(v,w)$ contains more triangles than the current version
of the synthetic graph, then $(v,w)$ is added and $(v',w')$ is removed.
The iteration stops when the number of inter-community triangles is greater than
or equal to that of the original graph.

Due to the removal of initially generated edges, the synthetic graph
may become disconnected. In this case, we apply an edge-swapping
post-processing step to reconnect every small connected component
to the main component (the connected component with the most nodes).
If the post-processing reduces the number of triangles,
we recall Algorithm~\ref{alg:gengraphtri}. The alternation between the post-processing
and Algorithm~\ref{alg:gengraphtri} is not guaranteed to yield a graph having exactly
the required number of triangles, so we stop the iteration
when the total number of triangles in the synthetic graph
is within a $98\%$ tolerance window with respect to the original one.

\subsection{Parameter Estimation for {\sf C-AGM}}
\label{ssec:calculate C-AGM}

We now discuss the methods for estimating the parameters
of a \mbox{\cagm} model from a given attributed graph
with a community partition $\CP$.

\subsubsection{Estimating $\Theta_M^c$}
\label{sssec:estimating-thetaMc}

The estimation of $\Theta_M^c$ reduces to computing the community-wise counters
that it relies on: intra- and inter-community degrees of every vertex,
the number of intra-community triangles for each community and the number
of inter-community triangles. As we mentioned in Section~\ref{sssec:CPGM},
degrees and triangle counts will be used to preserve global structural properties
of the generated graphs such as degree distribution and clustering coefficients.
They can be efficiently computed in the original graph both exactly
and under differential privacy.

\subsubsection{Estimating $\Theta^c_X$}
\label{sssec:estimating-thetacX}

In order to keep the estimation procedure tractable, we introduce the assumption
that attributes are independent. This assumption simplifies the estimation
and handles the sparsity of the attribute vectors when the number
of attributes is large. As seen in~\cite{PMFNG14,JYC16},
not having such an assumption severely limits the number of features
that can be practically handled.
Furthermore, as we will see in Section~\ref{sec:DP-calculation},
in addition to tractability, this assumption will also allow us to limit
the amount of noise added by the differentially private computation.

We will denote by $x_{\ell}$ be the value for the $\ell$-th component
of the attribute of vector $x$. Likewise, we will denote by $\tau_\ell(v)$
the value of the $\ell$-th component of the vector labelling vertex~$v$.
We estimate the probability that a node $v$ is labelled
with an attribute vector $x$ by the following formula:
\[
\Pr(x\mmid v,\Theta_X^c,\CP) =\Pr(x\mmid\psi_\ca{C}(v),\Theta_X^c)=
\prod_{\ell=1}^k \mbox{Pr}_\ell(x_{\ell}\mid\Theta^c_X,\psi_\ca{C}(v)),
\]
where $k$ is the number of columns of $X$ (ergo the cardinality
of all attribute vectors) and $\Pr_\ell(x_{\ell}\mid\Theta_X^c,\psi_\CP(v)) =
\frac{\mid\{v'\in\psi_\CP(v)\mid
\tau_\ell(v')=x_{\ell}\}\mid}{\mid\psi_\CP(v)\mid}$.

\subsubsection{Estimating $\Theta_F^c$}
\label{sssec:estimating-thetaFc}

As we discussed in Section~\ref{ssec:C-AGM}, for defining $\Theta_F^c$
it is necessary to define an aggregator function for pairs of attribute vectors.
Our aggregator function is based on the widely used cosine similarity,
that is, the cosine of the angle between two vectors.
Since the range $\ca{B}$ of aggregator functions needs to be discrete,
we split the range $[0,1]$ of the cosine similarity into a set of intervals,
determined by a parameter $\delta$ satisfying $0<\delta\le 1$.
Let $s_{\it cos}(x, x')$ denote the similarity between vectors $x$ and $x'$.
Our aggregator function is defined
as $\beta(x,x') = \left\lfloor\frac{s_{\it cos}(x,x')}{\delta}\right\rfloor$.
Note that, according to this definition,
$\ca{B}=\{\lfloor\frac{s}{\delta}\rfloor\mid s\in[0,1]\}$.
Finally, the probability of the attribute vectors of a pair of connected vertices
being described by an aggregated feature~$u\in\ca{B}$ is computed as 
\[
\Pr(\beta(\tau(v_i), \tau(v_j))=u\mmid\Theta_F^c,\CP, A_{i,j}=1)=
\begin{cases}
\frac{\mid\{ (v_p,v_q)\in\ca{E}\mid
\beta(\tau(v_p),\tau(v_q))=u\land\psi(v_p)=\psi(v_q)=
\psi(v_i)\}\mid}{\mid\{(v_p, v_q)\in\ca{E}\mid\psi(v_p)
=\psi(v_q)=\psi(v_i)\}\mid}\\
\hfill\mbox{\bf if~} \psi(v_i)=\psi(v_j);\\
\frac{\mid\{(v_p,v_q)\in\ca{E}\mid
\beta(\tau(v_p),\tau(v_q))=u\land \psi(v_p)\neq\psi(v_q)\}\mid}
{\mid\{(v_p, v_q)\in\ca{E}\mid\psi(v_p)\neq\psi(v_q)\mid}\\
\hfill\mbox{\bf if~}\psi(v_i)\neq\psi(v_j).
\end{cases}
\]

Compared to the approach introduced in~\cite{PMFNG14,JYC16}, 
our method uses a coarser granularity for aggregated features.
Thanks to that, it avoids the need to compute $2^{2k}$ different values, 
which is not only inefficient, but also results in an excessive amount 
of noise injected when applying differential privacy.

\section{Differentially Private {\sf C-AGM}}
\label{sec:DP-calculation}

In this section, we describe in detail our mechanisms for obtaining
differentially private instances of the \cagm{} model,
as well as the necessary adaptations of the sampling methods
when the model has been computed under differential privacy. 
As we discussed in Section~\ref{sec:preliminaries}, the difference 
between different instantiations of differential privacy for graphs 
lies in the definition of the pairs of graphs that are considered 
to be neighbouring datasets. Here, we adopt the following definition 
from~\cite{JYC16}. 

\begin{definition}[Neighbouring attributed graphs \cite{JYC16}]
\label{def-neigh-attr}
A pair of \mbox{attributed} graphs $G=(\ca{V},\ca{E},X)$ 
and $G'=(\ca{V},\ca{E}',\ca{X}')$ 
are \emph{neighbouring}, denoted $G\sim_{at}G'$, if and only if
they differ in the presence of exactly one edge or the attribute vector
of exactly one node. That~is,
$$G\sim_{at}G'\iff\lvert\ca{E}\nabla\ca{E}'\rvert=\! 1\ \lor 
(\exists_{v\in\ca{V}}\
\tau_G(v)\neq\tau_{G'}(v) \land \forall_{v'\in\ca{V}\setminus\{v\}}\
\tau_G(v)=\tau_{G'}(v)
).$$
\end{definition}

Definition~\ref{def-neigh-attr} entails that the existence 
of relations, that is the occurrence of edges, and the attributes 
describing every particular user, are treated as sensitive.  
On the contrary, vertex identifiers are treated as non-private.  
These criteria are in line with the current privacy policies 
of most social networking sites, where the fact that a profile exists 
is public information, but users can keep their personal information 
and friends list private or hidden from the general public. 
With Definition~\ref{def-neigh-attr} in mind, we describe in what follows
the differentially private computation of every parameter of \cagm.

\subsection{Obtaining the Community Partition}
\label{ssec:dp-community detection}

Our differentially private community partition method extends 
the algorithm ModDivisive~\cite{NIR16}, in such a way 
that it takes node attributes into account.
In its original formulation, \mbox{ModDivisive} searches for a community partition 
that maximises \emph{modularity}, a structural parameter encoding 
the intuition that a user tends to be more connected to users 
in the same community than to users in other communities \cite{NG04}. 
Modularity is defined as 
$$\sum_{C\in\ca{C}}\left(\frac{\ell_C}{m}-\left(\frac{d_C}{2m}\right)^2\right),$$ 
where $\ell_C$ is the number of edges between the nodes in $C$
and $d_C$ is the sum of degrees of the nodes in $C$.
ModDivisive uses the exponential mechanism, considering
the set of possible partitions as the categorical co-domain,
and using modularity as the scoring function.

In order to integrate node features into ModDivisive, we introduce 
a new objective function that combines the original modularity 
with an attribute-based quality criterion. The new objective function 
is defined as 
$$Q(\ca{C}) = w_s\cdot Q_s(\ca{C}) + w_a\cdot Q_a(\ca{C}),$$
where $w_s\in[0,1]$, $w_a=1-w_s$, $Q_s(\ca{C})$ is the modularity
of the original graph and $Q_a(\ca{C})$ is the modularity of an auxiliary graph
obtained from the original as follows. First, we take the vertex set
of the original graph. Then, we compute all pair\-wise similarities
between their associated feature vectors.
Similarities are computed using the cosine measure
(as done in Section~\ref{sssec:estimating-thetaFc}
for computing aggregated attributes, but without applying
the discretisation). Finally, we add to the auxiliary graph
the edges corresponding
to the $\left\lceil \frac{n(n-1)}{20}\right\rceil$
most similar attributed node pairs.

It is proven in~\cite{NIR16} that the global sensitivity of $Q_s(\ca{C})$
is upper bounded by $\frac{3}{m}$, where $m$ is the minimum
number of edges of all potential graphs to publish.
In the worst case, $\Delta_{Q_s(\ca{C})}=3$,
considering that the original graph is an arbitrary non-empty graph.
However, this is not the case for real-life social graphs,
so introducing more realistic assumptions about the value of $m$ allows us 
to use smaller values of $\Delta_{Q_s(\ca{C})}$ and thus reduce 
the amount of noise added in differentially privately computing $Q_s(\ca{C})$. 
Throughout this paper, we assume $m=10,000$, 
which leads to $\Delta_{Q_s(\ca{C})}=0.0003$. As we will see
in Section~\ref{sec:experiment}, all datasets used in our experiments
comply with this assumption. In what follows, we apply an analogous reasoning
for bounding $\Delta_{Q_a(\ca{C})}$. 

\begin{proposition}
\label{prop-sens-attr-modularity}
Every graph $G$ of order $n$ satisfies $LS_{Q_a(\ca{C})}(G)~\le~\frac{60}{n}.$
\end{proposition} 

\begin{proof}
Let $G\attrneigh G'$ be two neighbouring attributed graphs
and let $G_a$ and $G_a'$ be the auxiliary graphs obtained from $G$ and $G'$,
respectively. If the difference between $G$ and $G'$ consists only in one edge,
then $G_a=G_a'$, so in what follows we will consider that $G$ and $G'$
differ in one attribute vector. Let $v$ be the (sole) vertex
such that $\tau_G(v)\ne\tau_{G'}(v)$. In the worst case, we have that,
for every $w\in \ca{V}\setminus\{v\}$, $(v,w)\in G_a$ and $(v,w)\notin G'_a$
(or \emph{vice versa}).
It was shown in~\cite{NIR16} that the modularities of two graphs differing
in one edge differ in up to $\frac{3}{m}$, where $m$ is the minimum
number of edges. Then, in the worst case 
we have $\LS_{Q_a}(G)\le\frac{3(n-1)}{m_a}$,
where $n$ is the order of $G$ and $G'$, and $m_a$ is the minimum number of edges
in auxiliary graphs. As we discussed in Sect.~\ref{ssec:dp-community detection},
$m_a\ge \frac{n(n-1)}{20}$, so $\LS_{Q_a}(G)\le\frac{60}{n}$.
The proof is thus completed.
\end{proof}

Combining the result in~\cite{NIR16}
with that of Proposition~\ref{prop-sens-attr-modularity},
we conclude that $LS_{Q(\CP)}(G)\le 0.0003\cdot w_s +
\frac{60}{\lvert\ca{V}(G)\rvert}\cdot w_a$ for every $G$ satisfying 
the aforementioned assumptions, and use this value as an upper bound 
for $\Delta_{Q_a(\ca{C})}$.

\subsection{Attribute Vector Distribution}

As discussed in Section~\ref{ssec:calculate C-AGM}, 
given a community partition $\CP$, in order to obtain the differentially
private estimation of $\Theta_X^c$ (denoted by $\olTheta_X^c$),
we need to compute the probability distribution of each attribute
for every community, i.e.,
$\Pr_\ell(\tau_\ell(v)\mmid\overline{\Theta}_X^c, v\in C)$,
for each $\ell\le k$ (where $k$ is the number of attributes)
and $C\in\CP$. Computing this probability reduces to computing the number
of nodes whose $\ell$-th attribute has value $1$,
which we denote by $n_X^{\ell,C}$.
Let $n_X^C$ be the sequence $\left(n_X^{1,C}, n_X^{2,C}, \ldots,n_X^{k,C}\right)$.
In order to obtain the differentially private sequence
$\overline{n}^C_X=\left(\overline{n}_X^{1,C}, \overline{n}_X^{2,C},
\ldots,\overline{n}_X^{k,C}\right)$, we add to each element in $n_X^C$
noise sampled from $Lap\left(\frac{k}{\varepsilon_X}\right)$,
where $\varepsilon_X$ is the privacy budget reserved for this computation
and $k$ is the global sensitivity of $n_X^C$, as shown in the next result.

\begin{proposition}
\label{prop-sens-Theta-X}
The global sensitivity of
$n_X^C=\left(n_X^{1,C}, n_X^{2,C},\ldots,n_X^{k,C}\right)$
is $k$.
\end{proposition}

\begin{proof}
Let $G\attrneigh G'$ be two neighbouring attributed graphs, let $C\subseteq\ca{V}$
be a community and let $n_X^C(G)$ and $n_X^C(G')$ be the instances
of $n_X^C$ in $G$ and $G'$, respectively.
If the difference between $G$ and $G'$ consists only in one edge,
then $n_X^C(G)=n_X^C(G')$, so in what follows we will consider that $G$ and $G'$
differ in one attribute vector. Let $v$ be the (sole) vertex
such that $\tau_G(v)\ne\tau_{G'}(v)$. If $v\notin C$, then $n_X^C(G)=n_X^C(G')$.
On the contrary, if $v\in C$, for every component $\ell\in\{1,\ldots,k\}$
such that $\tau_{\ell,G}(v)\neq\tau_{\ell,G'}(v)$,
we have that $\big\lvert n_X^{\ell,C}(G)-n_X^{\ell,C}(G')\big\rvert=1$. 
In consequence, we have $\Delta_{n_X^C}=\displaystyle\max_{G\attrneigh G'}
\left\Vert n_X^C(G)-n_X^C(G')\right\Vert_1=k$.
\end{proof}

\subsection{Attribute-Edge Correlations}
\label{ssec:dp-att-edge-corr}

Recall that the aggregator function $\beta$ defined
in Section~\ref{sssec:estimating-thetaFc} maps every pair of attribute vectors
to a non-negative integer
in $\ca{B}=\{\lfloor\frac{s}{\delta}\rfloor\mid s\in[0,1]\}$,
for some $\delta\in [0, 1]$.
In order to estimate $\Theta_F^c$, we need to count, for each possible output
of $\beta$, the number of edges whose end-nodes are mapped to this value.
For every $t\in\ca{B}$ and every $C\in\CP$,
let $n_F^{t,C}$ be the number
of intra-community edges $(v, v')$ in $C$ such that $\beta(\tau(v), \tau(v'))=t$.
Likewise, let $n_{F\itbtw}^{\ell}$ be the number of inter-community edges
$(v,v')$ such that
$\beta(\tau(v), \tau(v'))=t$.
Thus, in order to compute $\olTheta_F^c$, we need to differentially privately
compute $n^{t, C}_F$ for every $t\in\ca{B}$ and every $C\in\CP$,
as well as $n_{F\itbtw}^t$ for every $t\in\ca{B}$.
We denote by $\oln_F^{t, C}$ and $\oln_{F\itbtw}^t$ the corresponding
differentially private values.

The global sensitivity of each of these sequences is $2(|\ca{V}|-2)$,
which is unbounded. To overcome this problem, we follow an approach analogous
to the one used in~\cite{JYC16} for counting attribute-edge correlations
in the entire graph. The method, introduced in~\cite{BBDS13},
consists in truncating the edge set of the graph to ensure that the degree
of all nodes is at most~$p$, which in the case of attribute-edge
correlations ensures that the global sensitivity is $2p$~\cite{JYC16,BBDS13}.
In consequence, for every $C\in\CP$, we obtain $\oln_F^{t,C}$ from $n_F^{t,C}$
by adding noise sampled from $Lap(\frac{2k}{\varepsilon_F})$,
where where $\varepsilon_F$ is the privacy budget reserved for this computation.
Likewise, we obtain $\oln_{F\itbtw}^t$ from $n_{F\itbtw}^t$
by adding noise sampled from $Lap(\frac{2k}{\varepsilon_F})$. 

\subsection{{\sf CPGM} Parameters}

In what follows we describe the computation of the parameters of \cpgm,
namely the set of intra-community and inter-community degrees and triangle counts.

\vspace{2mm}\noindent
{\bf Intra- and inter-community degrees.}
Following the trend of previous differentially private degree sequence
computation methods \cite{KS12,HLMJ09},
we first add noise to the raw values and then apply a post-processing
on the perturbed degree sequences to restore certain properties
of the original sequence, namely graphicality and the order of the nodes
in terms of their degrees, as well as certain community-specific properties.

Let $d_\itin^C=(d_\itin^{1,C},d_\itin^{2,C}\ldots, d_\itin^{m, C})$,
where $m=|C|$ and $d_\itin^{i,C}\le d_\itin^{i+1,C}$ ($1\le i < |C|$),
be the list of non-decreasingly ordered original intra-community degrees
in $C\in\CP$. Analogously,
let $d_\itbtw^C=(d_\itbtw^{1,C}, d_\itbtw^{2,C},\ldots, d_\itbtw^{m,C})$
be the sequence of inter-community degrees of nodes in $C$.

The global sensitivity of the degree sequence of the entire graph
is $2$, as adding or removing one edge changes the degrees of exactly
two nodes by $1$~\cite{HLMJ09}. The same is true for every $d_\itin^C$
and $d_\itbtw^C$, since the degrees of at most two intra-community nodes
(or at most one node in $C$ and one node outside of $C$) change by $1$.
Thus, for every $C\in\overline{\CP}$, we obtain from $d_\itin^C$
the differentially private sequence $\overline{d}_C$ by adding noise sampled
from $Lap(\frac{2}{\varepsilon_d})$ to every degree value. Similarly,
we obtain from $d_\itbtw^C$ the differentially private sequence
$\overline{d}_\itbtw^{i,C}$.
Afterwards, the noisy sequences are post-processed to restore three properties:
(i) the non-decreasing order between the intra-community degrees inside every
community, (ii) the graphicality of the intra-community degrees of every
community, and (iii) the graphicality of the inter-community degrees of all nodes
in the graph. Property (i) is enforced using the method proposed
in~\cite{HLMJ09}, whereas properties (ii) and (iii) are enforced
using the method proposed in~\cite{KS12}.

\vspace{2mm}\noindent
{\bf Numbers of intra- and inter-community triangles.}
The global sensitivity of the number of triangles in a graph 
is proven in~\cite{NRS07} to be $n-2$, where $n$ is the number of vertices.
The following result characterises the global sensitivity 
of the number of intra-community triangles.

\begin{proposition}
\label{prop-sens-no-intra-comm-triangles}
The global sensitivity of the number of intra-community triangles
of a graph $G$ with a community partition $\CP$
is $\Delta_\nintratriangle=\max_{C\in\CP}\{|C|-2\}$.
\end{proposition}

\begin{proof}
Let $G\attrneigh G'$ be two neighbouring attributed graphs
and let $\CP$ be a community partition of $\ca{V}$.
Let $n_\triangle^{\itin}(G)$ and $n_\triangle^{\itin}(G')$ be the numbers
of intra-community triangles of $G$ and $G'$, respectively.
If the difference between $G$ and $G'$ consists only in one attribute vector,
then $n_\triangle^{\itin}(G)=n_\triangle^{\itin}(G')$, so in what follows
we will consider that $G$ and $G'$ differ in one edge. We will assume,
without loss of generality, that $\ca{E'}\setminus\ca{E}=(v,v')$.
Two cases are possible for $\psi_\CP(v)$ and $\psi_\CP(v')$:
\begin{enumerate}[{\rm (i)}]
\item $\psi_\CP(v)\neq\psi_\CP(v')$. In this case, since $(v,v')$
is an inter-community edge, $n_\triangle^{\itin}(G)=n_\triangle^{\itin}(G')$.
\item $\psi_\CP(v)=\psi_\CP(v')=C$. In this case, we have that
$n_\triangle^{\itin}(G')-n_\triangle^{\itin}(G)=
|C\cap\ca{N}_G(v)\cap\ca{N}_G(v')|$,
that is the number of common neighbours of $v$ and $v'$ in the same community.
\end{enumerate}
It is simple to see that every pair of vertices $v$ and $v'$
such that $\psi_\CP(v)=\psi_\CP(v')=C$ satisfy
$|C\cap\ca{N}_G(v)\cap\ca{N}_G(v')|\le |C|-2$.
Hence, $\Delta_{n_\triangle^{\itin}}=\max_{C\in\CP}\{|C|-2\}$.
\end{proof}

Since the global sensitivity of triangle count queries is unbounded, 
the Laplace mechanism cannot be applied in this case. 
An accurate differentially private method for counting the number of triangles
of a graph is presented in~\cite{ZCPSX15}. This method uses the exponential
mechanism. It interprets the triangle count query as a categorical query, 
whose co-domain is a partition $\ca{O}$ of $\mathbb{Z}^+$. One of the elements 
of $\ca{O}$ is a singleton set composed exclusively of the correct output 
of the query (the correct number of triangles in this case), 
whereas every other element contains a set of incorrect values 
which are treated as equally useful. They define the notion 
of \emph{ladder function}, which is used as a scoring function 
on the elements of $\ca{O}$. A~ladder function gives better scores 
to the sets of values that are closer to the correct query answer. 
In order to differentially privately compute the number of triangles 
of a graph $G$, the ladder function approach starts by obtaining 
the correct number of triangles. Then, the ladder function is built, 
and a set $O\in\ca{O}$ is sampled following the exponential mechanism. 
Finally, a random element of $O$ is given as the differentially private 
output of the query~\cite{ZCPSX15}. 

It is shown in~\cite{ZCPSX15} that the best ladder function, 
in the sense that it adds the minimum necessary amount of noise, 
is the so-called \emph{local sensitivity at distance $t$}~\cite{NRS07}, 
which is denoted as $\LS_q(G,t)$, and is defined as the maximum local sensitivity 
of the query $q$ among all the graphs at edge-edit distance at most~$t$ from $G$. 
Formally, 
$\LS_q(G,t)=\max_{\{G'\;|\;\phi(G,G')\le t\}}\LS_q(G')$, where $\phi(G,G')$ 
is the edge-edit distance between $G$ and $G'$, that is the minimum number 
of edge additions and removals that transform $G$ into $G'$. 
It is also shown in~\cite{NRS07,ZCPSX15} that 
$\LS_q(G,t)=\max_{1\le i<j\le|\ca{V}|}\LS^q_{ij}(G,t)$, 
where $\LS^q_{ij}(G,t)=\max_{\{G',G''\;|\; \phi(G,G')\le t, 
G'\sim_{ij}G''} \left\vert q(G')-q(G'')\right\vert$ 
and $G'\sim_{ij}G''$ indicates that $G'$ and $G''$ differ in exactly 
the addition or removal of $(v_i,v_j)$. 

Here, we apply the ladder function approach for computing the number 
of intra-community triangles. To that end, 
we characterise the function $\LS_{\nintratriangle}((G,\CP),t)$
for every graph $G$ with a community partition $\CP$. 

\begin{proposition}
\label{prop-local-sens-intra-triangles}
For every graph $G$, every community partition $\CP$ of $G$, 
and every positive integer $t\ge 1$, 
$$\LS_{\nintratriangle}((G,\CP),t)=\max_{\{i,j\;|\;\psi_\CP(v_i)=\psi_\CP(v_j)\}}
\left\{\min\left\{
a_{ij}+\left\lfloor\frac{t+\min\{t,b_{ij}\}}{2}\right\rfloor,|\psi_\CP(v_i)|-2
\right\}\right\},$$
where
$a_{ij}=\lvert\{v_\ell\in\psi_{\CP}(v_i)\mid
A_{i,\ell}=1 \land A_{j,\ell}=1\}\rvert$ and
$b_{ij}=\lvert\{v_\ell\in\psi_{\CP}(v_i) \mid
A_{i,\ell}\oplus A_{j,\ell} =1\}\rvert$.
\end{proposition}

\begin{proof} 
Consider a graph $G$ with a community partition $\CP$, and a positive integer 
$t\ge 1$. As discussed in~\cite{NRS07,ZCPSX15}, 
$$\LS_{\nintratriangle}((G,\CP),t)=\max_{1\le i<j\le|\ca{V}|}
\LS^{\nintratriangle}_{ij}((G,\CP),t).$$ 
For every $i$ and $j$ such that $\psi_\CP(v_i)\ne\psi_\CP(v_j)$, 
we have that no intra-community triangle is created (resp. destroyed) 
by the addition (resp. removal) of $(v_i,v_j)$, 
so $\LS^{\nintratriangle}_{ij}((G,\CP),t)=0$. 
Thus, $$\LS_{\nintratriangle}((G,\CP),t)=
\max_{\{i,j\;|\;\psi_\CP(v_i)=\psi_\CP(v_j)\}}
\LS^{\nintratriangle}_{ij}((G,\CP),t).$$
 
We now focus on determining $\LS^{\nintratriangle}_{ij}((G,\CP),t)$ 
for every $i$ and $j$ such that $\psi_\CP(v_j)=\psi_\CP(v_j)=C$.
Consider a pair of such values $i$ and $j$, 
and let $S_1=\{v_\ell\in C\mid A_{i,\ell}=1 \land A_{j,\ell}=1\}$ 
and $S_2=\{v_\ell\in C\mid A_{i,\ell}\oplus A_{j,\ell} =1\}$\footnote{Note 
that $S_1$ and $S_2$ are the sets whose cardinalities define 
$a_{ij}$ and $b_{ij}$, respectively}. 

Let $\ca{G}_t$ be the class of all graphs that can be obtained 
by modifying $G$ as follows: 
\begin{enumerate}[1.]
\item Add $\min\{b_{ij},t\}$ arbitrary edges of the form $(x,y)$, 
where $x\in\{v_i,v_j\}$ and $y\in S_2$. 
\item If $t>b_{ij}$, take an arbitrary subset $S_3$ 
of vertices of $C\setminus(S_1\cup S_2\cup\{v_i,v_j\})$,  
with \mbox{cardinality} $\min\left\{\left\lfloor\frac{t-b_{ij}}{2}\right\rfloor, 
|C\setminus(S_1\cup S_2\cup\{v_i,v_j\})|\right\}$. For every $x\in S_3$, 
add the edges $(v_i,x)$ and $(v_j,x)$. 
\end{enumerate}

From the definition of $\ca{G}_t$, it follows that every \mbox{$G'\in\ca{G}_t$} 
satisfies $\phi(G,G')\le t$ and the graph $G''\sim_{ij}G'$ satisfies 
$$\vert\nintratriangle(G')-\nintratriangle(G'')\vert=
\min\left\{a_{ij}+\left\lfloor\frac{t+\min\{t,b_{ij}\}}{2}\right\rfloor,
|C|-2\right\}.$$

Now, consider an arbitrary graph $G'$, obtained by modifying $G$, 
such that $\phi(G,G')\le t$ and $G'\notin\ca{G}_t$. 
Also consider the graph $G''\sim_{ij}G'$. 
According to the definition of $\ca{G}_t$, the following situations are possible: 
\begin{enumerate}[{\rm (i)}]
\item $G'$ is the result of adding to $G$ a proper subset of the set of edges 
added by steps 1 and 2 of the procedure described above for obtaining 
an element of $\ca{G}_t$. 
In this case, only a proper subset of the triangles created (resp. destroyed) 
by the addition (resp. removal) of $(x,y)$ is added (resp. removed). 
Thus, $$\vert\nintratriangle(G')-\nintratriangle(G'')\vert<
\min\left\{a_{ij}+\left\lfloor\frac{t+\min\{t,b_{ij}\}}{2}\right\rfloor,
|C|-2\right\}.$$
\item $G'$ is the result of applying $t-t'$ additional modifications ($t'<t$) 
on an element $H$ of $\ca{G}_t$. Note that, by the definition 
of edge-edit distance, the additional modifications do not consist in reverting 
any edge addition made in steps 1 and 2 of the procedure described above. 
In this case, none of the additional modifications can result in the addition 
of a par of edges of the form $(v_i,x)$ and $(v_j,x)$, with $x\in C$, 
so \[\begin{aligned}
\vert\nintratriangle(G')-\nintratriangle(G'')\vert&=
\vert\nintratriangle(H)-\nintratriangle(H')\vert\\
&=\min\left\{a_{ij}+\left\lfloor\frac{t+\min\{t,b_{ij}\}}{2}\right\rfloor,
|C|-2\right\},
\end{aligned}\]
where $H'\sim_{ij}H$. 
\item In every other case, the transformation that allows 
to obtain $G'$ from $G$ can be divided into a set of edge additions 
as the ones described in (i) and a set of additional modifications 
as the ones described in (ii). Applying an analogous reasoning, 
we have that $$\vert\nintratriangle(G')-\nintratriangle(G'')\vert<
\min\left\{a_{ij}+\left\lfloor\frac{t+\min\{t,b_{ij}\}}{2}\right\rfloor,
|C|-2\right\}.$$
\end{enumerate}

Summing up the set of cases analysed above, we have that, for every 
$i$ and $j$ such that $\psi_\CP(v_i)=\psi_\CP(v_j)$, 
\[\begin{aligned}
\LS^{\nintratriangle}_{ij}((G,\CP),t)&=
\max_{\{G',G''\;|\;\phi(G,G')\le t, G'\sim_{ij}G''\}} 
\{\vert\nintratriangle(G')-\nintratriangle(G'')\vert\}\\
&=\min\left\{a_{ij}+\left\lfloor\frac{t+\min\{t,b_{ij}\}}{2}\right\rfloor,
|\psi_\CP(v_i)|-2\right\}
\end{aligned}\]
and, in consequence, 
\[\begin{aligned}
\LS_{\nintratriangle}((G,\CP),t)&=
\max_{\{i,j\;|\;\psi_\CP(v_i)=\psi_\CP(v_j)\}}
\LS^{\nintratriangle}_{ij}((G,\CP),t)\\
&=\max_{\{i,j\;|\;\psi_\CP(v_i)=\psi_\CP(v_j)\}}\left\{\min\left\{
a_{ij}+\left\lfloor\frac{t+\min\{t,b_{ij}\}}{2}\right\rfloor,|\psi_\CP(v_i)|-2
\right\}\right\}.
\end{aligned}\]
The proof is thus completed.
\end{proof}

In Proposition~\ref{prop-local-sens-intra-triangles}, the operator $\oplus$ 
denotes exclusive or. Notice that $\LS_{\nintratriangle}((G,\CP),t)$ 
can be efficiently computed for small values 
of $t$, and it converges to the efficiently computable global sensitivity 
$\Delta_\nintratriangle$ for $t\ge 2\max_{C\in\CP}|C|$, so it can be used 
for efficiently and privately computing the number of intra-community triangles. 

Finally, for computing the number of inter-community triangles, 
we use the method from~\cite{ZCPSX15} to compute the number of triangles 
of the entire graph, and subtract from it the number of intra-community 
triangles computed with the method described in this subsection. 

\subsection{Summary}
\label{ssec:sampling-dp-graphs}

The methods discussed in the previous subsections allow to compute
a differentially private instance of \cagm{}.
In what follows, we will use the notation \cagmdp{}
to clearly distinguish differentially private instances of \cagm.
The privacy budget $\varepsilon$
is split among the different computations as follows:
$\varepsilon_c=\frac{\varepsilon}{2}$ for the community partition method,
$\varepsilon_F=\frac{\varepsilon}{6}$ for the estimation
of~$\overline{\Theta}_F^c$, and $\varepsilon_d=\varepsilon_\triangle=
\varepsilon^\itin_\triangle=\varepsilon_X=\frac{\varepsilon}{12}$
for the estimation of degree distributions, triangle counts and $\olTheta_X^c$.

\begin{remark}
\label{prop-dp-full-method}
Parameter estimation for \cagmdp{} satisfies
$(\varepsilon_c+\varepsilon_X+\varepsilon_F+ \varepsilon_d+
\varepsilon_\triangle+\varepsilon_\triangle^\itin)$-differential privacy.
\end{remark}

\begin{proof}
The result follows straightforwardly from the fact that
$\varepsilon_c+\varepsilon_X+\varepsilon_F+ \varepsilon_d+
\varepsilon_\triangle+\varepsilon_\triangle^\itin=\frac{\varepsilon}{2}+
\frac{\varepsilon}{12}+\frac{\varepsilon}{6}+ \frac{\varepsilon}{12}+
\frac{\varepsilon}{12}+\frac{\varepsilon}{12}=\varepsilon$
\end{proof}

\section{Experiments}
\label{sec:experiment}

The purpose of our experiments is to empirically validate 
the following two claims: (i) our \cpgm{} model outperforms existing models
in generating graphs whose community structures are more similar to those
of the input graphs without sacrificing the ability to preserve
global structural properties, and (ii) differentially private instances 
of \cagm{} also outperform preceding models in terms of the preservation 
of community structure, while remaining comparable in terms of the preservation 
of global structural properties. 

\subsection{Datasets}

For our experiments, we use three real-world social networks with node attributes.
The first one has been collected from \mbox{Petster},
a website for pet owners to communicate~\cite{K13}.
It is an undirected graph whose nodes represent hamster owners.
Each node is labelled with attributes containing information
about the user's pet. We extracted 13 binary attributes
from 8 categorical attributes such as favourite food, gender, colour, species,
year of birth, etc. The second one is a subset of Facebook available
via SNAP \cite{snapnets}. In this dataset, node attributes are already binary
and are tagged with serial pseudonyms.
For our experiments, we selected the first 50 attributes
with the smallest serial numbers.
Finally, the third dataset, Epinions, is a directed graph extracted
from an online consumer reviews system, where every vertex
represents a reviewer~\cite{MA07}.
In the original dataset, a directed edge from node $A$ to node $B$ exists
if user $A$ trusts  the reviews of $B$.
For our experiments, we derived an undirected graph from the original dataset
by keeping the same vertex set and adding an undirected edge for every pair
of mutually trusting users. Additionally, we selected the 50 most
frequently rated products as node attributes. If the user rated the product,
the value is set to $1$, otherwise it is set to $0$.
Table~\ref{tab:datasets} summarises the main statistics of the three datasets.

\begin{table}[!ht]
\centering
\begin{tabular}{l|c|c|c|c|c}
\hline
{\bf Dataset}&{\bf \#nodes}&{\bf \#edges} & {\bf \#$\triangle$}\quad\quad
& {\bf cl. coeff.} & {\bf\#attr.}\\
\hline\hline
Petster & 1,898 &12,534 &16,750  &0.14 & 13\\
Facebook & 3,953 & 84,070 & 1,526,985 &  0.54&50 \\
Epinions &29,515 &106,147 & 235,790 & 0.13& 50\\ \hline
\end{tabular}
\caption{Datasets used for our experiments.}\label{tab:datasets}
\end{table}

\subsection{Evaluation Measures}
\label{ssec:eval-measures}

For every pair $(G,G')$, where $G$ is a real-life graph and $G'$ is a synthetic
graph sampled from a model learned from $G$, we evaluate the extent
to which $G'$ preserves the following properties of $G$.

\vspace{1mm}\noindent
{\bf Numbers of edges and triangles:}
Our evaluation measures in this case are the relative errors of the numbers
of edges and triangles in $G'$ with respect to those in $G$.
We define these measures as
$\rho_E=\frac{\bigl\lvert\lvert\ca{E}_{G'}\rvert-\lvert\ca{E}_{G}\rvert\bigr\rvert}
{\lvert\ca{E}_{G}\rvert}$
and $\rho_\triangle=\frac{\lvert
n_{\triangle}(G')-n_{\triangle}(G)\rvert}{n_{\triangle}(G)}$, respectively.

\vspace{1mm}\noindent
{\bf Global clustering coefficient:}
The \emph{global clustering coefficient} (GCC) of a graph measures the proportion
of wedges, that is, paths of length $2$, that are embedded in triangles.
It is defined as $\frac{3n_\triangle}{n_w}$, where $n_w$ is the number of wedges
and $n_\triangle$ is the number of triangles. We compare $G$ and $G'$
in terms of the relative error of the GCC of $G'$ with respect to that of $G$.
We denote this measure by $\rho_c$.

\vspace{1mm}\noindent
{\bf Degree distribution.}
We compare $G$ and $G'$ in terms of the \emph{Hellinger distance}
between their degree distributions.
The Hellinger distance has been deemed as the most appropriate distance
for comparing probability distributions in previous works
on graph synthesising~\cite{JYC16,MPS13}.
Given two probability distributions $p_1$ and $p_2$ on a discrete domain $W$,
the Hellinger distance between $p_1$ and $p_2$ is defined as
$$H(p_1, p_2) = \frac{1}{\sqrt{2}}\sqrt{\sum_{w\in W}
(\sqrt{p_1(w)}-\sqrt{p_2(w)})^2}.$$
The Hellinger distance yields values in the interval $[0,1]$. The more similar
two distributions are, the smaller the Hellinger distance between them.
For the particular case of degree distributions, we compute $p_d$ and $p'_d$,
which are defined on the domain $W=\{0, 1, \ldots, n-1\}$,
where $n$ is the number of vertices in $G$ and $G'$.
For every $i\in W$, $p_d(i)$ (resp. $p'_d(i)$) is the probability that a vertex
of $G$ (resp. $G'$) has degree $i$. The final score used for comparing $G$ and $G'$
is $H_d=H(p_d,p'_d)$.

\vspace{1mm}\noindent
{\bf Local clustering coefficients.}
In a graph $G$, the \emph{local clustering coefficient} (LCC)
of a node $v$ measures the proportion of pairs of mutual neighbours of $v$
that are connected by an edge. In the context of social graphs,
$LCC(v)$ is an indicator of the likelihood of $v$'s mutual friends
to also be friends. $LCC(v)$ is defined
as $\frac{2\sum_{v_i,v_j\in\ca{N}(v)}A_{i,j}}
{\lvert\ca{N}(v)\rvert\cdot(\lvert\ca{N}(v)\rvert-1)}$
where $\ca{N}(v)$ is the set of $v$'s neighbours.
For comparing $G$ and $G'$ in terms of local clustering coefficients,
we compute the distributions $p_{\ell c}$ and $p'_{\ell c}$, which are defined
in the domain $W=\{c\;|\;\exists_{v\in\ca{V}}(LCC_G(v)=c \lor LCC_{G'}(v)=c)\}$
in such a way that for every $i\in W$, $p_{\ell c}(i)$ (resp. $p'_{\ell c}(i)$)
is the probability that a vertex of $G$ (resp. $G'$) has LCC $i$.
We compare $G$ and $G'$ in terms of $H(p_{\ell c},p'_{\ell c})$,
and denote this measure as $H_{\ell c}$.

\vspace{1mm}\noindent
{\bf Distribution of attribute-edge correlations.}
Recall that $k$ represents the number of components of every attribute vector
labelling the vertices of both $G$ and $G'$. Given a community partition
$\CP$ of $G$, we define for every $C\in\CP$ the distributions $p_F^C$
and $\widetilde{p}_F^C$ in the domain $W=\{0,1\}^k$ in such a way that,
for every $i\in W$, $p_F^C(i)$ (resp. $\tilde{p}_F^C(i)$) is the probability
that a vertex belonging to $C$ in $G$ (ergo, in the context of this paper,
also in $G'$) is labelled with the attribute vector $i$ in $G$ (resp. in~$G'$).
We compare $G$ and $G'$ in terms of the parameter $\rho_a$, which is defined
as $\rho_a = \max_{C\in\CP}\left\{H(p_F^C, \widetilde{p}_F^C)\right\}.$

\vspace{1mm}\noindent
{\bf Detectability of community partition.}
We evaluate to what extent state-of-the-art community detection algorithms
find similar communities in $G$ and $G'$.
To that end, we use the averaged $F_1$ score,
denoted Avg-$F_1$, of the community structures $\CP$ and $\CP'$ determined
by the algorithm in $G$ and $G'$, respectively.
The averaged $F_1$ score has been widely used for evaluating community detection
algorithms~\cite{YML13,YL13,NIR16}.
Given two communities $C_1$ and $C_2$, the $F_{1}$ score between
these two communities, denoted $F_1(C_1,C_2)$ combines two auxiliary measures:
\emph{precision} and \emph{recall}.
Precision is defined as
${\it prec}(C_1, C_2)=\frac{\lvert C_1\cap C_2\rvert}{\lvert C_1\rvert}$,
whereas recall is defined
as ${\it recall}(C_1, C_2)=\frac{\lvert C_1\cap C_2\rvert}{\lvert C_2\rvert}$.
Precision and recall are combined
as $F_1(C_1,C_2) = \frac{2\cdot{\it prec}(C_1,C_2)
\cdot{\it recall}(C_1,C_2)} {{\it prec}(C_1,C_2)+{\it recall}(C_1,C_2)}$.
If both precision and recall are zero, $F_1$ is made zero by convention.
Following the evaluation strategy introduced in~\cite{YML13,YL13,NIR16},
given two sets of communities $\ca{C}_1$ and $\ca{C}_2$, we first determine
the average of the $F_1$ values between every community of $\CP_1$
and its best match in $\CP_2$ (in terms of $F_1$), then the average
of the $F_1$ values between every community of $\CP_2$ and its best match
in $\CP_1$, and finally these two values are also averaged.
The average $F_{1}$-score is defined as
$$\frac{1}{2\lvert\ca{C}_1\rvert}\sum_{C_i^1\in\ca{C}_1}\max_{C_j^2\in\ca{C}_2}
F_1(C_i^1, C_j^2) + \frac{1}{2\lvert\ca{C}_2\rvert}\sum_{C_i^2\in\ca{C}_2}
\max_{C_j^1\in\ca{C}_1} F_1(C_i^2, C_j^1).$$
Avg-$F_1$ values are in the interval $[0,1]$.
The larger the value of Avg-$F_1$, the more similar
the community structures of $\mathcal{C}_1$ and $\mathcal{C}_2$
are considered to be.

\subsection{Results and Discussion}
\label{ssec:results}

We first evaluate the ability of our new edge generation model, \cpgm,
to synthesise graphs that preserve the community structures of the original
graphs along with global structural properties. Then, we assess
the overall quality of the differentially private \cagmdp{} model.

\subsubsection{Evaluation of {\sf CPGM}}
\label{sssec:results-cpgm}

We compare \cpgm{} with the two most similar counterparts reported
in the literature: \tricycle~\cite{JYC16} and \dcsbm~\cite{KE11}.
\mbox{\tricycle{}} has been shown to preserve a number of global structural
properties, but it does not aim to preserve the community structure;
whereas \dcsbm{} belongs to a family of models that has been shown
to preserve the community structure, disregarding global structural properties.
Figure~\ref{fig:f1-CPGM} shows the extent to which the community structures
found by the state-of-the-art community detection algorithm Louvain~\cite{BGLL08}
in the synthetic graphs generated by each model are similar to those detected
in the corresponding original graphs. Table~\ref{tab:commfeat}
compares the behaviour of all three edge generation models in terms of global
structural properties. In all cases, the values shown are averaged
over 10 executions.

\begin{figure}[h]
\centering
\includegraphics[scale=0.5]{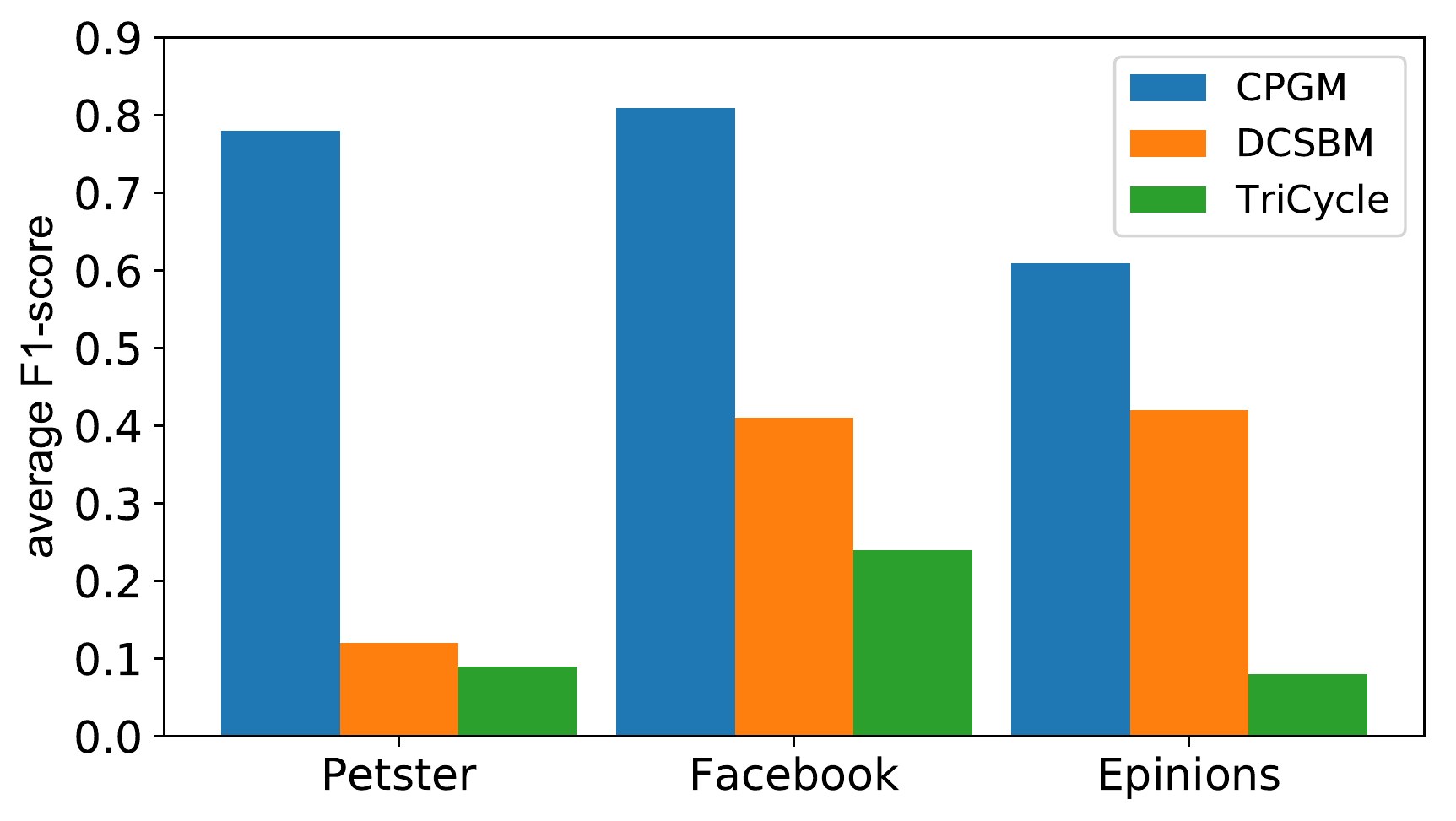}
\caption{Similarities of community structures found by Louvain
in synthetic graphs to those found in the original graphs.}
\label{fig:f1-CPGM}
\end{figure}

From the analysis of these results, we extract three major observations.
Firstly, as Figure~\ref{fig:f1-CPGM} shows,
the community structures of the graphs synthesised using our \cpgm{} model
are consistently more similar to those of the original graphs, in comparison
to those induced by \dcsbm{} and \tricycle. This supports our claim
that \cpgm{} is able to preserve community structure to a larger extent.
Note that, in several cases,
our model performs almost twice as good as the second best, \dcsbm.
As~expected, \tricycle{} shows the poorest results, corroborating
the intuition that community structure needs to be explicitly included
in the generative model if we want synthetic graphs to preserve it.
The previous observations support our design choices of preserving
(i) the community structure, and (ii) differentiated intra- and inter-community
structural properties.

Secondly, the graphs generated by our \cpgm{} model are consistently
the most accurate in terms of the distributions of local clustering coefficients
(see right-most column of Table~\ref{tab:commfeat}),
and are the most accurate in terms of global clustering coefficient
in all but one dataset (see column labelled $\rho_c$ in Table~\ref{tab:commfeat}).
An analogous observation can be made
for the number of triangles (column labelled $\rho_\triangle$).
We consider that these observations support 
our design choice of preserving separate intra- and inter-community
edge densities and triangle counts. The comparably poorer performance
of \dcsbm{} in terms of global and local clustering coefficients
also corroborates the need to explicitly model them, as do \cpgm{}
and \mbox{\tricycle}. A more detailed graphical description of the behaviour
of the three models in terms of the distributions of local clustering coefficients
is shown in Figure~\ref{fig:degreeAndLcc} (d)--(e),
in Appendix~\ref{app:distribution}. The figure shows the complementary
cumulative distribution functions of local clustering coefficients
for the three models, and highlights the special ability of our \cpgm{} model
to capture the behaviour of the distribution for denser-than-normal
graphs (the Facebook dataset in this case).

\begin{table}[!ht]
\centering
\begin{tabular}{l|c|c|c|c|c|c}
\hline
\textbf{Dataset}& \textbf{Model} & \textbf{$\rho_E$} &
\textbf{$\rho_\triangle$} & \textbf{$\rho_c$} & \textbf{$H_d$} &
\textbf{$H_{\ell c}$} 
\\
\hline\hline
\multirow{3}{*}{\textbf{Petster}}  &
\cpgm& 0.00& 0.18& 0.05 & 0.16& 0.19
\\
& \dcsbm & 0.00& 0.12& 0.49& 0.17& 0.27
\\
& \tricycle & 0.00& 0.00& 0.19& 0.18& 0.21 
\\ \hline
\multirow{3}{*}{\textbf{Facebook}} & \cpgm & 0.00 & 0.03& 0.32& 0.15& 0.32
\\
& \dcsbm & 0.00& 0.25& 0.71& 0.25& 0.64
\\
& \tricycle & 0.00& 0.04& 0.56& 0.37& 0.60 
\\\hline
\multirow{3}{*}{\textbf{Epinions}} & \cpgm &
0.001& 0.04& 0.27& 0.13&0.31
\\
& \dcsbm & 0.002& 0.60&0.83&0.14&0.26
\\
& \tricycle &0.001& 0.04& 0.22& 0.10& 0.31
\\ \hline
\end{tabular}
\caption{Comparison of edge generative models 
in terms of global structural properties.}
\label{tab:commfeat}
\end{table}

Finally, we point out that all three models successfully preserve
the properties of the degree distribution.
Our \cpgm{} model produces the most consistent results in all the datasets
in terms of both Hellinger distances and the shape of the complementary
cumulative distribution functions (see Figure~\ref{fig:degreeAndLcc} (a)--(c)
in Appendix~\ref{app:distribution}). Again, \cpgm{} is the one that better
captures the cumulative distribution for the denser Facebook dataset.

Summing up, the results shown in this subsection support our claim
that synthetic graphs sampled from our \cpgm{} model preserve
the community structure of the original graph
to a considerably larger extent than its closest
counterparts, without sacrificing the ability to preserve global structural
properties. Additionally, these results show that the manner in which \cpgm{}
computes intra- and inter-community parameters also helps it outperform
competing models in preserving local and global clustering coefficients.

\subsubsection{Evaluation of differentially private {\sf C-AGM}}
\label{sssec:eval-g-acmdp}

We compare \cagmdp{} with two other models. The first one
is the differentially private \agm{} model using \tricycle{}
as edge set generator~\cite{JYC16}. We refer to this model
as \agmdptri. It was shown in~\cite{JYC16} that, despite the noise added
to guarantee privacy, \agmdptri{} still preserves to some extent
\tricycle's ability to capture global structural properties.
As we saw in the previous section, \tricycle{} performs poorly
in preserving the community structure of the original graph,
so we additionally consider for our evaluation an additional model,
which is a modification of \cagmdp{} where \cpgm{} is replaced
by \dcsbm{} as the edge generation model. We refer
to this model as \mbox{\cagmdpd}.

\begin{figure*}[t]
\centering
\includegraphics[scale=0.27]{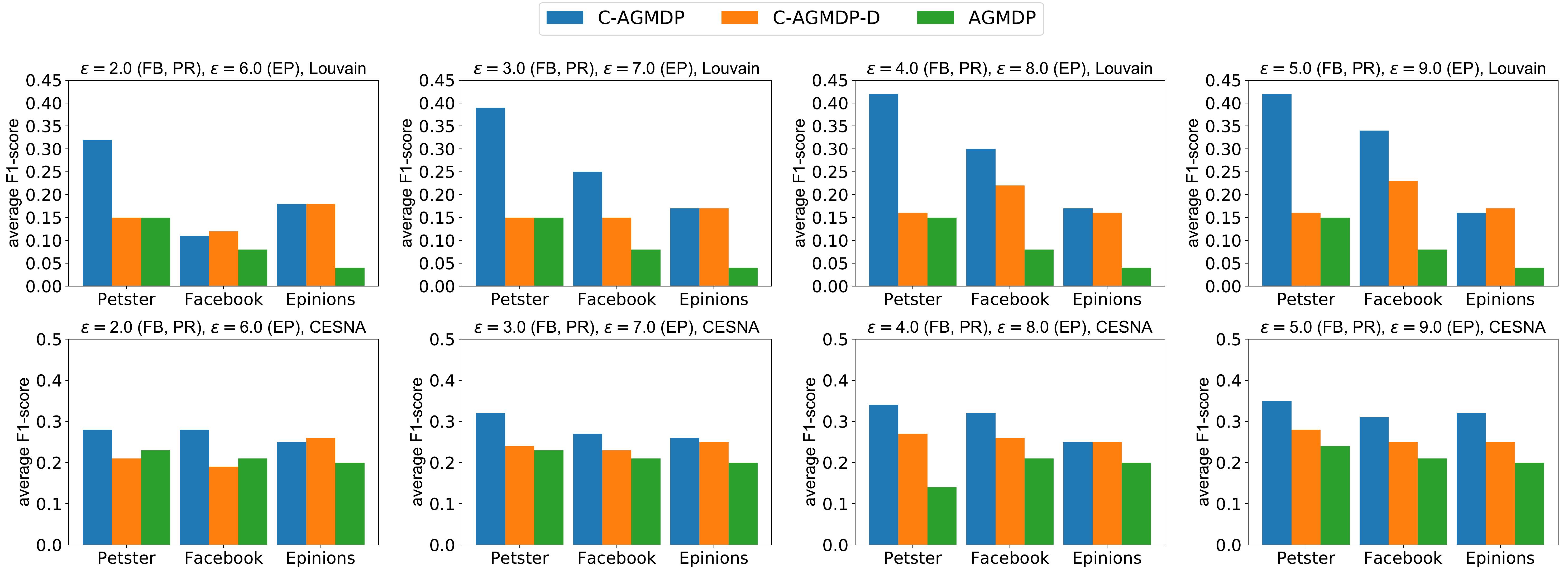}
\caption{Comparison of differentially private models
in terms of community structure preservation.}\label{fig:f1DP}
\end{figure*}

In our experiments, we compare the behaviour of the three models under four
privacy budgets. For the smaller Petster and Facebook datasets,
we use the values $2.0$, $3.0$, $4.0$ and $5.0$.
For the considerably larger Epinions dataset, using these values results
in considerably noisy outputs, where pairwise differences between the models
have been blurred away by the noise. Thus, in order to enable comparisons,
we show for this dataset the results obtained using the values
$6.0$, $7.0$, $8.0$ and $9.0$.
Every particular comparison of the three models uses a common privacy budget.
Since \cagmdpd{} and \agmdptri{} each have less parameters
than \cagmdp, we re-allocate in each case the remaining privacy budget
to other computations. In estimating \cagmdpd, we allocate to
community partition the same budget as for \cagmdp,
that is $\frac{\varepsilon}{2}$. \cagmdpd{} requires to compute the numbers
of edges between every pair of communities. We assign to this computation
the budget used in \cagmdp{} for counting the number of intra-community
triangles, i.e. $\frac{\varepsilon}{12}$. Finally, \cagmdpd{} is given
for degree sequence computation the budget
$\varepsilon'_d=\varepsilon_d+\varepsilon_\triangle=\frac{\varepsilon}{6}$,
as it does not require to count the global number of triangles.
Since \mbox{\agmdptri{}} computes every parameter computed by \cagmdp,
except for the community partition (which takes
half of the budget of \cagmdp) we double the budget assigned to every
other computation of \agmdptri.
Notice that these re-allocations give some advantages to \cagmdpd{}
and \agmdptri{} in their comparison with our \cagmdp{} model,
as they will be able to more accurately compute some of the parameters
they have in common. We chose to allow this advantage considering
that requiring a smaller number of computations is in fact a positive
feature of a differentially private method, which should not be punished
in the comparison. In \cagmdp{} and \mbox{\cagmdpd},
ModDivisive is run with $w_s=0.98$. Finally, in estimating the distribution 
of attribute-edge correlations (as discussed in Section~\ref{ssec:dp-att-edge-corr}), 
we set the maximum degree parameter $p$ to $100$. 

Figure~\ref{fig:f1DP} displays the behaviours of the three models in terms
of community structure preservation, whereas Tables~\ref{tab:dppetster},
\ref{tab:dpfacebook} and~\ref{tab:dpepinions} summarise their behaviours
in terms of global structural properties and attribute-edge correlations
on Petster, Facebook and Epinions,
respectively. In what follows, we analyse these results from three
different perspectives.

\begin{table}[!th]
\centering
\begin{tabular}{l|l|l|l|l|l|l|l}
\hline
\multicolumn{1}{c|}{$\varepsilon$} & \textbf{Model} & \multicolumn{1}{c|}{$\rho_E$}
&\multicolumn{1}{c|}{$\rho_\triangle$} & \multicolumn{1}{c|}{$\rho_c$}
&\multicolumn{1}{c|}{$H_d$} & \multicolumn{1}{c|}{$H_{\ell c}$} & $\rho_a$ \\
\hline
\hline
\multirow{3}{*}{2.0} & \cagmdp & 0.56 & 0.09 & 0.43 & 0.42 & 0.39 & 0.14 \\
&\cagmdpd & 0.22 & 0.55 & 0.69 & 0.30 & 0.44 & 0.23 \\
&\agmdptri & 0.25 & 0.09 & 0.25 & 0.23 & 0.30 & 0.17 \\
\hline
\multirow{3}{*}{3.0} & \cagmdp & 0.31 & 0.09 & 0.23 & 0.33 & 0.29 & 0.16 \\
&\cagmdpd & 0.11 & 0.30 & 0.55 & 0.24 & 0.34 & 0.16 \\
&\agmdptri & 0.13 & 0.09 & 0.21 & 0.19 & 0.26 & 0.17 \\
\hline
\multirow{3}{*}{4.0} & \cagmdp & 0.19 & 0.08 & 0.20 & 0.29 & 0.26 & 0.13 \\
&\cagmdpd & 0.06 & 0.29 & 0.54 & 0.22 & 0.32 & 0.14 \\
&\agmdptri & 0.09 & 0.08 & 0.19 & 0.19 & 0.25 & 0.17 \\
\hline
\multirow{3}{*}{5.0} & \cagmdp & 0.13 & 0.08 & 0.08 & 0.24 & 0.23 & 0.09 \\
&\cagmdpd & 0.04 & 0.29 & 0.54 & 0.20 & 0.31 & 0.12 \\
&\agmdptri & 0.06 & 0.10 & 0.17 & 0.18 & 0.24 & 0.16\\
\hline
\end{tabular}
\caption{Comparison of differentially private models on Petster.}
\label{tab:dppetster}
\end{table}

\vspace{2mm}
{\bf Community structure preservation.}
In Figure~\ref{fig:f1DP}, the four uppermost charts display the extent
to which the community structures found by Louvain in the synthetic graphs
generated by each differentially private model are similar to those detected
in the corresponding original graphs. Two important features
of the Louvain algorithm are shared by ModDivisive,
the method used for obtaining community partitions in \cagmdp{}
and \cagmdpd. Both generate a community partition,
and both operate by maximising modularity. In order to assess whether
the community structures induced by our models in the synthetic graphs
are also detectable by algorithms based on different criteria,
we additionally obtained analogous results using
the algorithm CESNA~\cite{YML13}. These results are shown
in the lowermost four charts of Figure~\ref{fig:f1DP}. Unlike Louvain,
CESNA takes node attributes into consideration
for computing communities. However, CESNA
tends to obtain substantially overlapping communities, whereas both
\cagmdp{} and \cagmdpd{} assume a partition.
CESNA requires as a parameter the number of communities,
which we set to~$10$.

\begin{table}[!th]
\centering
\begin{tabular}{l|l|l|l|l|l|l|l}
\hline
\multicolumn{1}{c|}{$\varepsilon$} & \textbf{model} & \multicolumn{1}{c|}{$\rho_E$}
&\multicolumn{1}{c|}{$\rho_\triangle$} & \multicolumn{1}{c|}{$\rho_c$}
&\multicolumn{1}{c|}{$H_d$} & \multicolumn{1}{c|}{$H_{\ell c}$} & $\rho_a$ \\
\hline
\hline
\multirow{3}{*}{2.0} & \cagmdp & 0.10 & 0.01 & 0.59 & 0.25 & 0.54 & 0.13 \\
&\cagmdpd & 0.02 & 0.90 & 0.90 & 0.18 & 0.86 & 0.08 \\
&\agmdptri & 0.06 & 0.13 & 0.58 & 0.31 & 0.59 & 0.19 \\
\hline
\multirow{3}{*}{3.0} & \cagmdp & 0.05 & 0.01 & 0.51 & 0.22 & 0.47 & 0.09 \\
&\cagmdpd & 0.01 & 0.89 & 0.89 & 0.16 & 0.90 & 0.06 \\
&\agmdptri & 0.03 & 0.15 & 0.59 & 0.32 & 0.59 & 0.19 \\
\hline
\multirow{3}{*}{4.0} & \cagmdp & 0.03 & 0.01 & 0.50 & 0.21 & 0.46 & 0.07 \\
&\cagmdpd & 0.01 & 0.87 & 0.88 & 0.15 & 0.88 & 0.05 \\
&\agmdptri & 0.02 & 0.15 & 0.60 & 0.32 & 0.19 & 0.08 \\
\hline
\multirow{3}{*}{5.0} & \cagmdp & 0.02 & 0.01 & 0.48 & 0.21 & 0.43 & 0.06 \\
 & \cagmdpd & 0.01 & 0.84 & 0.88 & 0.16 & 0.86 & 0.04 \\
 & \agmdptri & 0.01 & 0.15 & 0.60 & 0.33 & 0.59 & 0.19\\\hline
\end{tabular}
\caption{Comparison of differentially private models on
Facebook.}
\label{tab:dpfacebook}
\end{table}

From the analysis of these results,
the most relevant observation is that, considering the communities detected
by both Louvain and CESNA, the graphs sampled
from our \mbox{\cagmdp}
model consistently rank as the ones whose community structure is most similar
to that of the corresponding original graphs. In the particular cases of Petster
and Facebook with the Louvain algorithm, the similarity values displayed
by our \cagmdp{} model are in some cases close to twice better
than their counterparts for \cagmdpd, and considerably more
than those for \agmdptri. An additional important observation
is that our model shows less variation for different amounts of noise,
when compared to \cagmdpd{} and \agmdptri.

\vspace{2mm}\noindent
{\bf Distributions of attribute-edge correlations.}
On each original graph, we compute the distributions of attribute-edge
correlations in each community detected by CESNA.
Then, we compute the equivalent distributions on each synthetic graph
and compare it to that of the original graph in terms of $\rho_a$
(see right-most columns of Tables~\ref{tab:dppetster}, \ref{tab:dpfacebook}
and~\ref{tab:dpepinions}).

From the analysis of these results,
we can see that the synthetic graphs sampled from \cagmdp{}
and \cagmdpd, the two models that consider community structures,
consistently outperform those sampled from \agmdptri{} in terms of $\rho_a$.
Another important observation is that the qualities, in terms of $\rho_a$,
of synthetic graphs sampled from our \cagmdp{} model and those sampled
from its variant \cagmdpd{} are quite similar. This observation suggests
that \cagmdp{} can in some cases be seen as a \emph{meta-model},
where several edge generation models can be used,
e.g. \cpgm{} and \dcsbm{} in this experiment.

\vspace{2mm}\noindent
{\bf Global structural properties.}
From the analysis of Tables~\ref{tab:dppetster}, \ref{tab:dpfacebook}
and~\ref{tab:dpepinions}, we can see that, as expected, our \cagmdp{} model
suffers a larger degradation than \cagmdpd{} and \agmdptri{}
in terms of the measures that depend on parameters
for which the latter models were allocated
larger privacy budgets, most notably the numbers of edges. This problem
was particularly serious on Petster, and became less important
as the number of edges of the real graph increased. It is worth noting, however,
that in the cases where a differentially private parameter underwent
a post-processing, most notably regarding the number of triangles,
our model obtained considerably better results. For example,
the error rate dropped to $0.01$ for the Facebook dataset.
Also in the Facebook graph,
despite the larger error rate in the number of edges,
in some cases our model showed roughly the same or even better performance
in preserving the degree sequence and clustering coefficients
than the \mbox{\agmdptri{}} model. Also note that, although \cagmdpd{}
performs best in preserving the degree sequences, in general
it failed to preserve the clustering coefficients. 

\begin{table}[!th]
\centering
\begin{tabular}{l|l|l|l|l|l|l|l}
\hline
\multicolumn{1}{c|}{$\varepsilon$} & \textbf{model} & \multicolumn{1}{c|}{$\rho_E$}
&\multicolumn{1}{c|}{$\rho_\triangle$} & \multicolumn{1}{c|}{$\rho_c$}
&\multicolumn{1}{c|}{$H_d$} & \multicolumn{1}{c|}{$H_{\ell c}$} & $\rho_a$ \\
\hline
\hline
\multirow{3}{*}{6.0} & \cagmdp & 0.08 & 0.20 & 0.53 & 0.23 & 0.21 & 0.09 \\
&\cagmdpd & 0.08 & 0.62 & 0.85 & 0.19 & 0.23 & 0.09 \\
&\agmdptri & 0.13 & 0.31 & 0.16 & 0.13 & 0.27 & 0.13 \\
\hline
\multirow{3}{*}{7.0} & \cagmdp & 0.07 & 0.18 & 0.53 & 0.23 & 0.21 & 0.09 \\
&\cagmdpd & 0.07 & 0.69 & 0.91 & 0.20 & 0.24 & 0.09 \\
&\agmdptri & 0.13 & 0.31 & 0.16 & 0.12 & 0.27 & 0.13 \\
\hline
\multirow{3}{*}{8.0} & \cagmdp & 0.05 & 0.16 & 0.54 & 0.20 & 0.20 & 0.09 \\
&\cagmdpd & 0.05 & 0.70 & 0.91 & 0.18 & 0.24 & 0.08 \\
&\agmdptri & 0.14 & 0.32 & 0.18 & 0.11 & 0.26 & 0.13 \\
\hline
\multirow{3}{*}{9.0} & \cagmdp & 0.04 & 0.14 & 0.54 & 0.17 & 0.20 & 0.09 \\
&\cagmdpd & 0.04 & 0.75 & 0.93 & 0.18 & 0.25 & 0.08 \\
&\agmdptri & 0.15 & 0.31 & 0.18 & 0.11 & 0.26 & 0.13\\
\hline
\end{tabular}
\caption{Comparison of differentially private models on Epinions.}
\label{tab:dpepinions}
\end{table}

\section{Conclusions}
\label{sec:conclusions}

We have presented, to the best of our knowledge, the first community-preserving 
differentially private method for publishing synthetic attributed graphs. 
To devise this method, we developed \cagm, a new community-preserving 
generative attributed graph model. 
We have equipped \cagm{} with efficient parameter estimation 
and sampling methods, and have devised differentially private 
variants of the former. A comprehensive set of experiments 
on real-world datasets support the claim that our method 
is able to generate useful synthetic graphs satisfying 
a strong formal privacy guarantee. 
Our main direction for future work is to improve \cagm{} 
by increasing the repertoire of community-related statistics  
captured by the model, and by equipping it with a new differentially private 
community partition method that integrates node attributes 
via a low-sensitivity objective function and/or 
differentially private maximum-likelihood estimation methods. 

\vspace*{.7cm}
\noindent \textbf{Acknowledgements:} The work reported in this paper 
received funding from Luxembourg's Fonds National de la Recherche (FNR), 
via grant C17/IS/11685812 (PrivDA).

\appendix

\section{Distribution of degree and local clustering coefficient}
\label{app:distribution}

Figure~\ref{fig:degreeAndLcc} shows the comparison of degree distributions (a--c) 
and distributions of local clustering coefficients (d--f) in terms of 
the \emph{complementary cumulative distribution functions} (CCDF). 
Every degree or LCC value (x-axis) is mapped to the percentage of vertices 
having a larger value (y-axis).

\begin{figure*}[!ht]
\subfigure[Petster]{\includegraphics[scale=0.35]{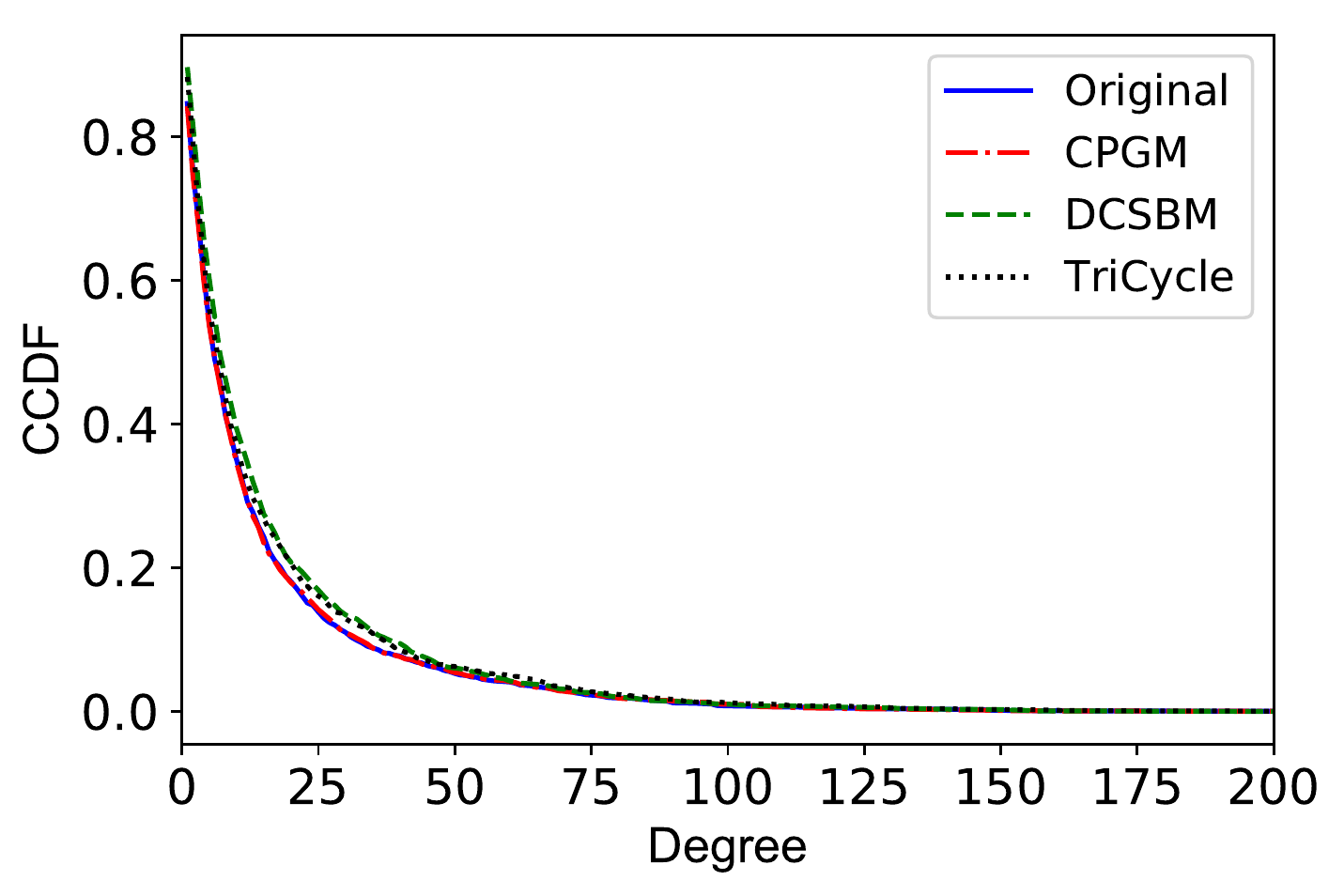}}
\subfigure[Facebook]{\includegraphics[scale=0.35]{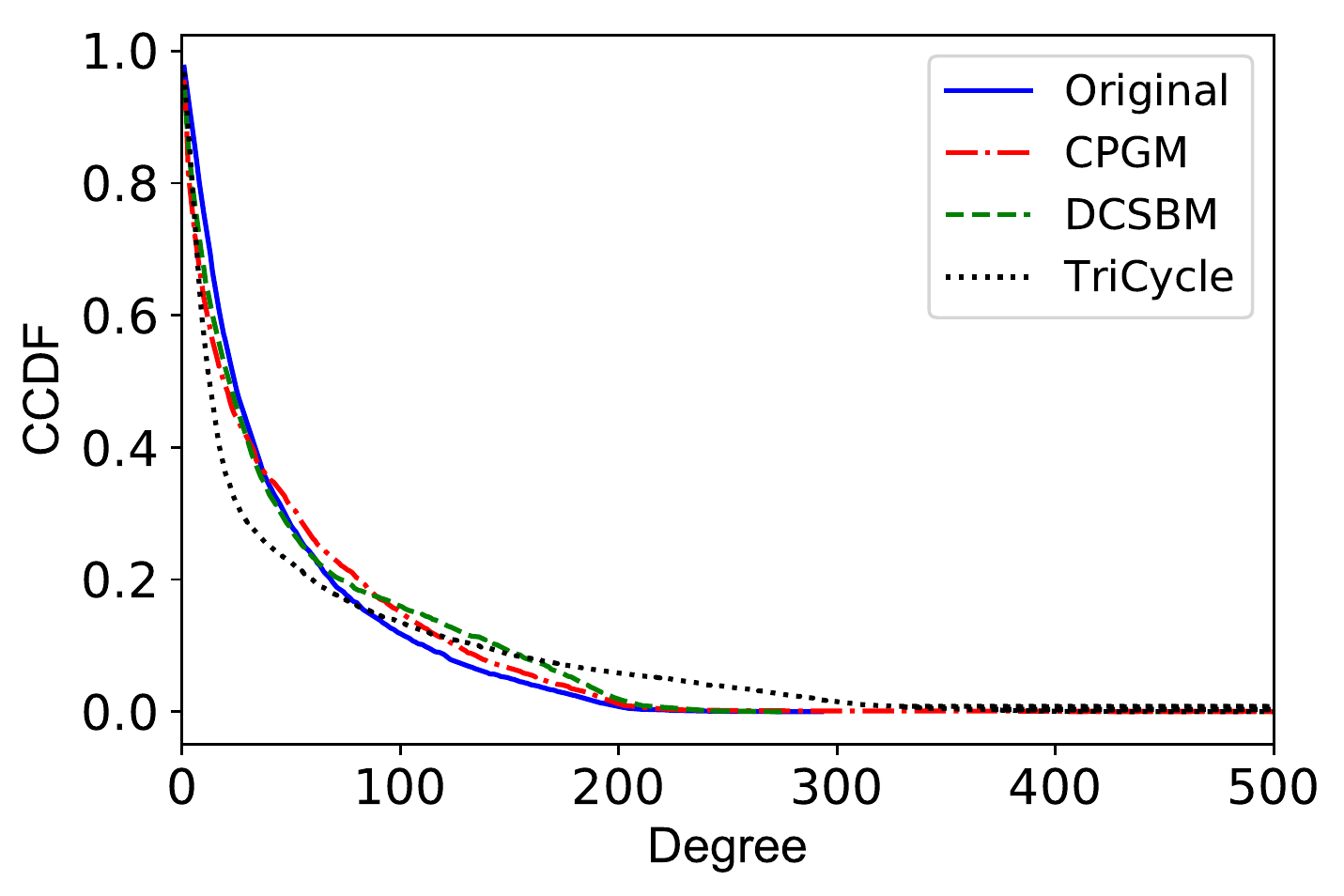}}
\subfigure[Epinions]{\includegraphics[scale=0.35]{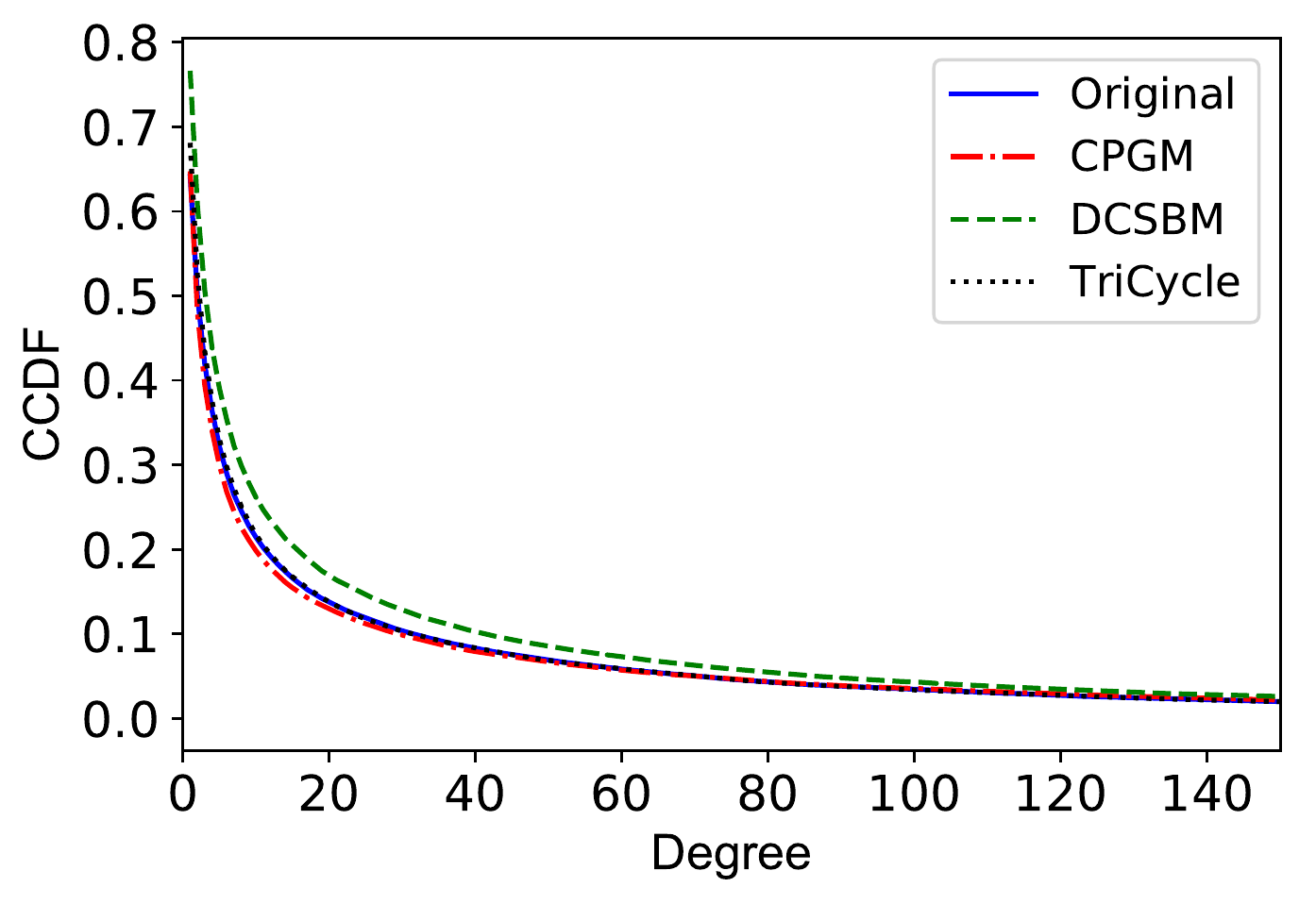}}
\\
\subfigure[Petster]{\includegraphics[scale=0.35]{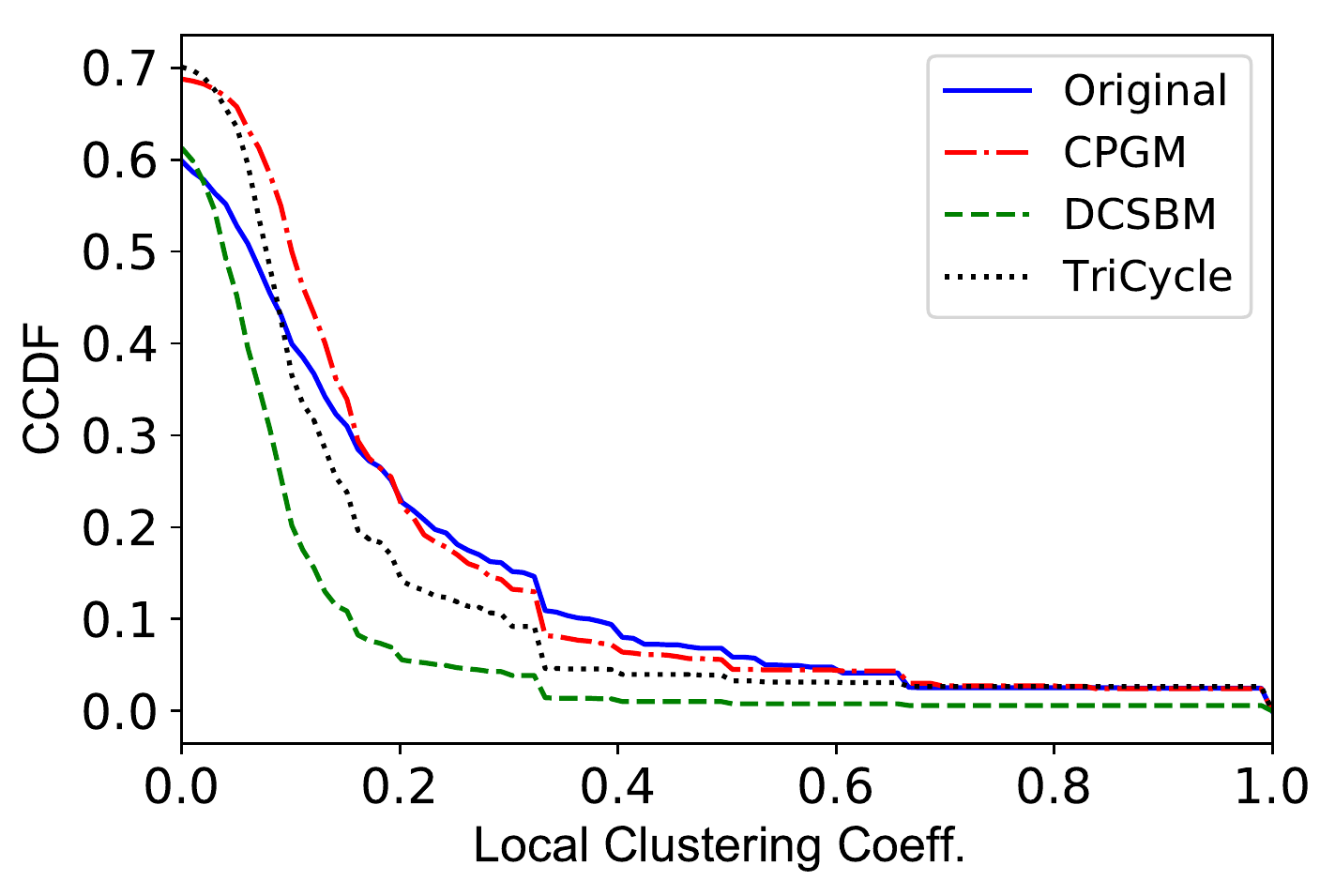}}
\subfigure[Facebook]{\includegraphics[scale=0.35]{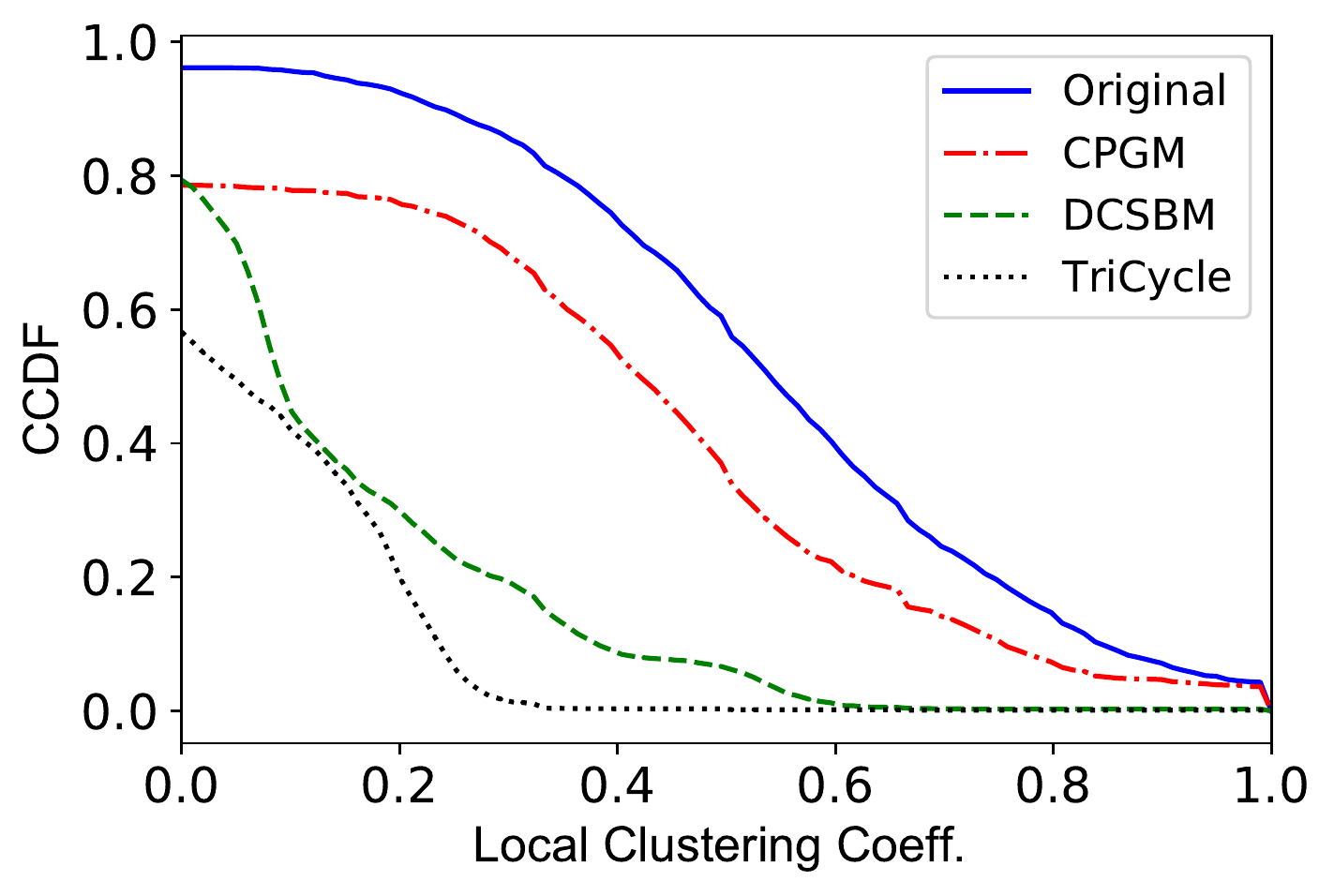}}
\subfigure[Epinions]{\includegraphics[scale=0.35]{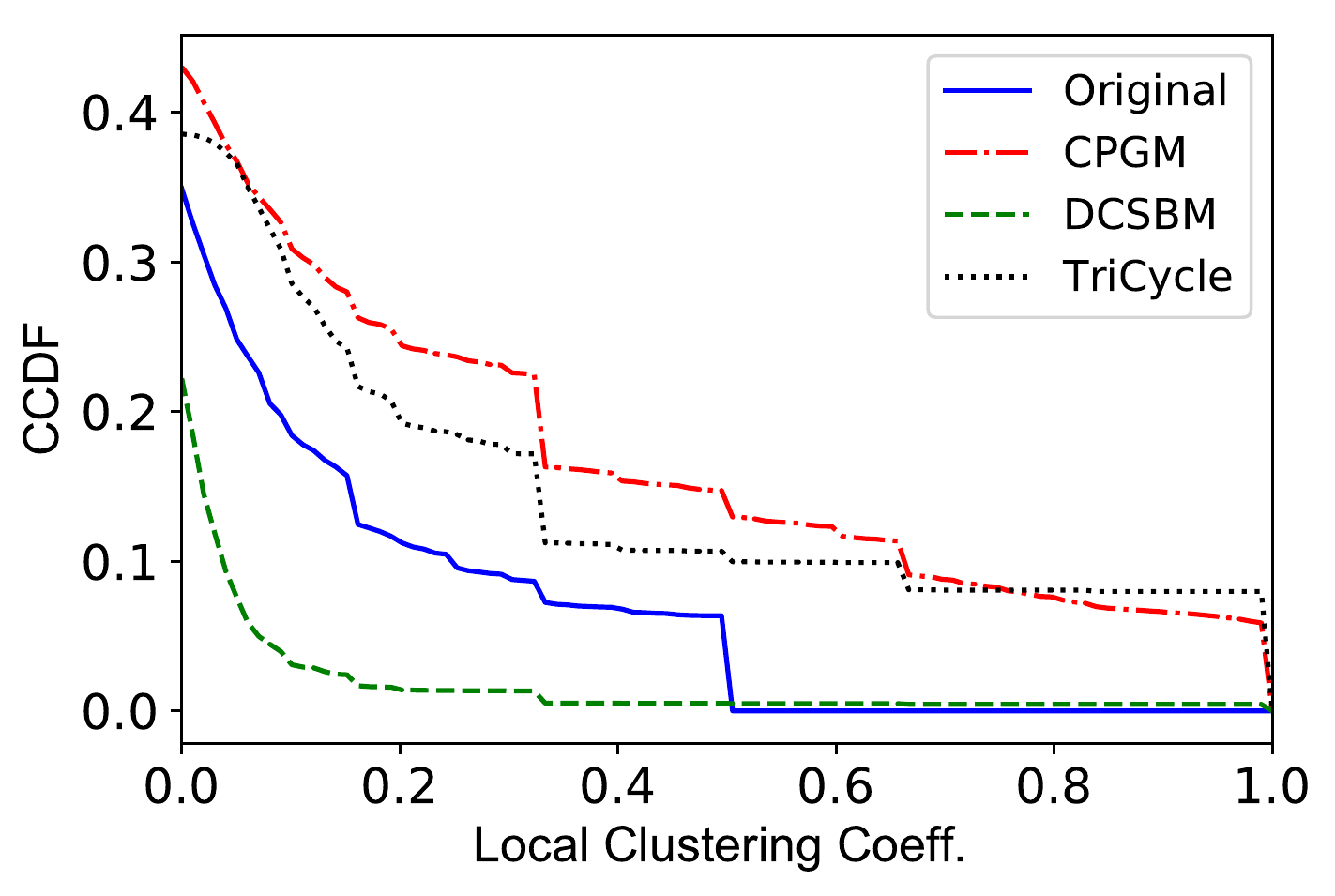}}
\caption{Detailed comparison of edge generative models 
in terms of complementary cumulative distribution functions 
of degree and local clustering coefficient distributions.}
\label{fig:degreeAndLcc} 
\end{figure*}

\end{document}